%% file: main.tex
\documentclass[10pt,journal,compsoc]{IEEEtran}

\usepackage[english]{babel}
\usepackage{blindtext}
\usepackage{times}
\usepackage{balance} 
\usepackage{amsmath,amsfonts,amssymb,amsthm}
\usepackage{algorithm}
\usepackage[noend]{algpseudocode}
\usepackage{graphicx}
\usepackage{booktabs}
\usepackage{xspace}
\usepackage{multirow}
\usepackage{url}
\usepackage{epsfig,subfigure}
\usepackage{epstopdf}
\usepackage{color}
\usepackage{upgreek}
\usepackage{url}
\usepackage{caption}
\usepackage{graphbox}

\newcommand{\squishlisttight}{
 \begin{list}{$\bullet$}
  { \setlength{\itemsep}{0pt}
    \setlength{\parsep}{0pt}
    \setlength{\topsep}{0pt}
    \setlength{\partopsep}{0pt}
    \setlength{\leftmargin}{2em}
    \setlength{\labelwidth}{1.5em}
    \setlength{\labelsep}{0.5em}
} }

\newcommand{\squishnumlist} {
\newcounter{qcounter}
\begin{list}{\arabic{qcounter}.~}{\usecounter{qcounter}} 
{  \setlength{\itemsep}{0pt}
    \setlength{\parsep}{0pt}
    \setlength{\topsep}{0pt}
    \setlength{\partopsep}{0pt}
    \setlength{\leftmargin}{2em}
    \setlength{\labelwidth}{1.5em}
    \setlength{\labelsep}{0.5em}
}}

\newcommand{\squishend}{
  \end{list}
}

\newcommand{\eat}[1]{}
\newcommand{\eatTKDE}[1]{}

\newcommand{\TKDE}[1]{#1} 
 
\newcommand{\revise}[1]{#1}

\newcommand{\kw}[1]{{\ensuremath {\mathsf{#1}}}\xspace}

\newcommand{\stitle}[1]{\vspace{1ex} \noindent{\bf #1}}
\long\def\comment#1{}

\newcommand{\blue}[1]{\textcolor{blue}{#1}}

\newcommand{\score}{\kw{f}}
\newcommand{\feq}{\kw{feq}}
\newcommand{\dis}{\kw{cor}}
\newcommand{\rep}{\kw{rep}}
\newcommand{\smy}{\kw{smy}}
\newcommand{\g}{\kw{g}}
\newcommand{\dist}{\kw{dist}}
\newcommand{\smyp}{\kw{kWTS}-\kw{problem}}
\newcommand{\kVDO}{\kw{kWTS}-\kw{problem}}
\newcommand{\kVDOmodel}{\kw{kWTS}}
\newcommand{\anc}{\kw{anc}}
\newcommand{\dec}{\kw{des}}
\newcommand{\des}{\kw{des}}

\newcommand{\greedy}{\kw{GTS}}
\newcommand{\FEQ}{\kw{FEQ}}
\newcommand{\AGG}{\kw{AGG}}
\newcommand{\CAGG}{\kw{CAGG}}
\newcommand{\Baseline}{\kw{Baseline}}

\newcommand{\nb}{\kw{N}}
\newcommand{\vtree}{\kw{Vtree}}
\newcommand{\DP}{\kw{OTS}}

\newcommand{\brute}{\kw{Brute}-\kw{Force}}

\newcommand{\na}{\kw{na}}
\newcommand{\LCA}{\kw{LCA}}
\newcommand{\LCAs}{\kw{LCAs}}
\newcommand{\LATT}{\kw{LATT}}
\newcommand{\LNUR}{\kw{LNUR}}
\newcommand{\ANIM}{\kw{ANIM}}
\newcommand{\IMAGE}{\kw{IMAGE}}
\newcommand{\YAGO}{\kw{YAGO}}
\newcommand{\CD}{\kw{CD}}
\newcommand{\ALD}{\kw{ALD}}
\newcommand{\WC}{\kw{WC}}
\newcommand{\TS}{\kw{HDS}}

\newcommand{\revision}[1]{\blue{#1}} 

\theoremstyle{definition}
\newtheorem{definition}{Definition} 
\newtheorem{example}{Example} 
\newtheorem{theorem}{Theorem} 
\newtheorem{lemma}{Lemma}

\begin{document}

\title{Efficient and Optimal Algorithms for Tree Summarization with Weighted Terminologies}
\author{Xuliang~Zhu,
        Xin~Huang, Byron~Choi, 
        Jianliang~Xu,  William~K.~Cheung, \\ 
        Yanchun~Zhang, and~Jiming~Liu \\

\IEEEcompsocitemizethanks{\IEEEcompsocthanksitem X. Zhu, X. Huang, B. Choi, J. Xu, W. Cheung, and J. Liu are with the Department of Computer Science, Hong Kong Baptist University, Hong Kong, China.\protect\\
E-mail: \{csxlzhu,xinhuang,bchoi,xujl,william,jiming\}@comp.hkbu.edu.hk
\IEEEcompsocthanksitem Y. Zhang is with the Guangzhou University, China and Victoria University, Australia.\protect\\
E-mail: Yanchun.Zhang@vu.edu.au
}
\thanks{(Corresponding author: Xin Huang.)}
}

\IEEEtitleabstractindextext{
\begin{abstract}

\input{abstract}

\end{abstract}

}

\maketitle   
\IEEEpeerreviewmaketitle

\input{tex/intro}

\input{tex/relate}
\input{tex/problem}
\input{tex/analysis}

\input{tex/greedy}
\input{tex/dp}
\input{tex/vtree}
\input{tex/exp}

\section{Conclusion and Future Work} \label{sec.con}
In this paper, we motivate and study the tree summarization problem to select $k$ representative vertices to summarize a weighted tree. We first propose an efficient greedy algorithm \greedy with quality guarantee. In addition, we develop an optimal algorithm \DP based on dynamic programming techniques to find exact answers in polynomial time.  We also propose an efficient tree reduction technique to improve efficiency of both \greedy and \DP. 
Extensive experiments on real-world datasets demonstrate the superiority of our proposed algorithms against state-of-the-art methods. This paper also opens up several interesting problems. One challenging direction is how to generate the node weights in a hierarchy for tree summarization. In the application of terminology search, the node weight is regarded as the occurrence of a certain terminology. However, users may not input an exact terminology every time. 
Such unmatching terminologies and alternative names desire to be resolved by string matching and semantic matching.

\ifCLASSOPTIONcompsoc
 \section*{Acknowledgments}
\else
 \section*{Acknowledgment}
\fi
This work is supported by HK RGC Grants Nos. 12200021, 12202221, 12201520,  22200320, 12201119, and 12201518.

\ifCLASSOPTIONcaptionsoff
  \newpage
\fi

\bibliographystyle{IEEEtran}
{
\bibliography{DAG,truss}
}

\balance

\end{document}

%% file: abstract.tex
Data summarization that presents a small subset of a dataset to users has been widely applied in numerous applications and systems. Many datasets are coded with hierarchical terminologies, e.g., gene ontology, disease ontology, to name a few.  
In this paper, we study the weighted tree summarization. We motivate and formulate our \kVDO as selecting a diverse set of $k$ nodes to \underline{s}ummarize a hierarchical \underline{t}ree $T$ with \underline{w}eighted terminologies. We first propose an efficient greedy tree summarization algorithm \greedy. It solves the problem with $(1-1/e)$-approximation guarantee. 
Although \greedy achieves quality-guaranteed answers approximately, but it is still not optimal. To tackle the problem optimally, we further develop a dynamic programming algorithm \DP to obtain optimal answers for \kVDO in  $O(nhk^3)$ time, where $n, h$ are the node size and height in  tree $T$. The algorithm complexity and correctness of \DP are theoretically analyzed. In addition, we propose a useful optimization technique of tree reduction to remove useless nodes with zero weights and shrink the tree into a smaller one, which ensures the efficiency acceleration of both \greedy and \DP in real-world datasets. 
Moreover, we illustrate one useful application of graph visualization based on the answer of $k$-sized tree summarization and show it in a novel case study. 
Extensive experimental results on real-world datasets show the effectiveness and efficiency of our proposed approximate and optimal algorithms for tree summarization. Furthermore, we conduct a usability evaluation of attractive topic recommendation on ACM Computing Classification System dataset to validate the usefulness of our model and algorithms.

%% file: tex/intro.tex
\section{introduction}\label{sec.intro}

A hierarchy data model is commonly used to depict the terminologies and their hierarchical relationships,  such as gene ontology~\cite{gene}, disease ontology~\cite{disease-ontology}, the International Classification of diseases-9~\cite{icd9}, medical subject heading, systematized nomenclature of medicine-clinical terms~\cite{SNOMED_CT}, and also the ACM Computing Classification System~\cite{ACMCCS}. 
Beside the topological structure of hierarchies, the terminologies are usually associated with vertex weights in a large number of real applications. For example, in biomedicine, the weight of terminologies obtained from literature search tools or electronic health records (EHR) are usually aggregated by events, such as the occurrences of diseases, and the number of search terms~\cite{jing2011graphical,jing2014complementary}; in academic, the  weight of a research topic terminology is the total number of papers published in this topic, e.g., the number of papers published in database venues. Such terminologies with hierarchical structures are often modeled as trees or directed acyclic graphs. Therefore, we consider \emph{a hierarchy data with weighted terminologies} as \emph{a weighted tree} throughout this paper.

However, real-world hierarchy data is often large-scale with numerous terminologies.
For instance, as of 2011, the  systematized nomenclature of medicine-clinical terms contains more than 311,000 medical concepts~\cite{SNOMED_CT}.
This brings significant difficulty for users to understand the essence of terminologies, even with the aid of a direct visualization tool. It is impossible to explore them interactively. Therefore, it desires to design efficient and effective algorithms for data summarization on such weighted hierarchies~\cite{jing2011graphical,wu2017finding,zhu2020top}, which gives a small-scale representation to summarize the whole dataset. A good summarization of weighted hierarchical data can benefit a wide range of applications such as summarized recommendation~\cite{zhu2020top}, visual data exploration~\cite{akoglu2012opavion, zhu2020hdag}, and snippet generation for information search~\cite{fakas2015diverse}.

\begin{figure}[t]
\small
\scriptsize
\centering
\includegraphics[width=1.01\linewidth]{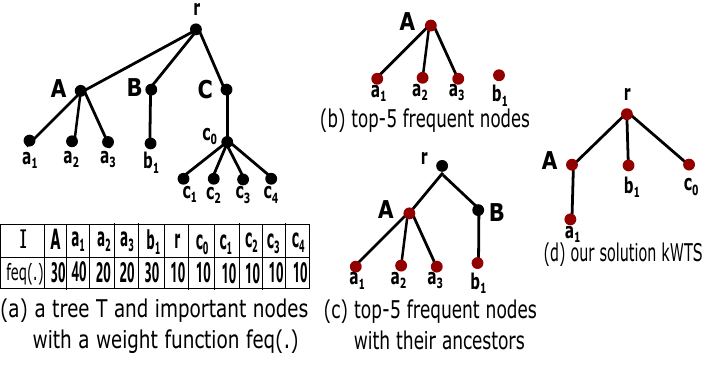}
\vspace{-0.6cm}
\caption{A running example of tree summarization with weighted terminologies.}
\label{fig.example}
\vspace{-0.5cm}
\end{figure}

We motivate our problem of tree summarization and illustrate one typical application of graph visualization on disease ontology with weighted terminologies. For instance, Figure \ref{fig.example}(a) shows one sample example of disease ontology. The nodes $r, A, a_1, ...$ represent disease terminologies. The edges represent the instance relationship, e.g., $(r, A)$ indicates that $A$ is an instance of $r$. In general, the disease (node $r$) includes mental health disease (node $A$), syndrome disease (node $B$), and cellular proliferation disease (node $C$). Furthermore, the diseases of cellular proliferation (node $C$) have one instance of cancer (node $c_0$). In the third level, the types of cancers (node $c_0$) can be categorized into cells (node $c_1$), organ systems (node $c_2$), and so on. Given a table of node weights that record the occurrence of diseases in a hospital (see the table in Figure \ref{fig.example}(a)), one may seek a summary report that presents a clear structure of frequent diseases.
 Obviously, if we show all diseases in the disease ontology, it is beyond the human cognition ability to distinguish any clear structure.  Thus, we consider how to select a small set of $k$ (e.g., $k=5$) important and representative elements to summarize the entire dataset. The simplest approach is to pick the most frequent elements. However, as this approach does not make use of hierarchical terminologies, we cannot see the inter-relationships between the selected elements in the resulted summary (see Figure \ref{fig.example}(b)).
An improved approach is to also include all the ancestors of the top-$k$ elements in the terminological structure (see Figure \ref{fig.example}(c)).  While this improved approach provides a more intuitive summary, it still suffers from two drawbacks. First, the summarization may lack diversity and miss specific but small groups (e.g., $c_1$, $c_2$, $c_3$, and $c_4$), which might yield inaccurate summarization for users. Second, similar elements are not summarized in a high-level concept. Moreover, to show all ancestors of frequent elements, a large summarization might be resulted, e.g., Figure \ref{fig.example}(c) has 7 nodes, which is greater than the given $k$. In contrast, Figure \ref{fig.example}(d) depicts a better summarization of the input dataset that describes four types of diseases (including $A$, $a_1$, $b_1$, and $c_0$), where element $a_1$ with the highest frequency represents a large proportion of type-$A$ diseases. This summarization in graph visualization offers direct, simplified, intuitive and human-friendly images to help users understand the overview of this analyzed disease dataset. 

In this paper, we investigate the problem of selecting a diverse set of vertices to summarize a weighted tree where vertices have non-negative weights. 
Formally, we formulate the \smyp, that is, given a tree $T$ with weighted terminologies and a positive integer $k$, finding a set of $k$ representative vertices to summarize the whole tree $T$ with the largest summary score. 
This new problem formulation is based on an objective function of summary score, taking into account the representativeness, diversity, and high-score coverage simultaneously. We provide a novel method of summarizing large tree datasets by reducing the original dataset to a manageable size. It intends to depict, highlight, and distinguish the important nodes and links within the hierarchal structure. To find high-quality summarized results, we propose a simple but efficient algorithm \greedy. \greedy is a greedy algorithm based on a well-designed greedy strategy that iteratively adds a representative vertex with the largest summary contribution for the overall summary score, until the answer has $k$ representative vertices. The greedy method can achieve at least $(1-1/e)$-approximation of the optimal answer in terms of our objective function.

Over the conference version \cite{huang2017ontology} of this manuscript, we further investigate exact solutions to the tree summarization problem. We propose a dynamic programming algorithm to achieve optimal answers in Section~\ref{sec.dp}. The motivation is that although the existing greedy method \greedy~\cite{huang2017ontology} achieves the quality-guaranteed answers, the answers are still not optimal. However, finding optimal answers brings significant challenges. An intuitive approach is to enumerate all possible summary sets to find the best answer, which may incur expensive computations. In fact, there exist exact polynomial-time algorithms to tackle \kVDO. Therefore, we propose an algorithm \DP based on dynamic programming to optimally solve the problem. The general idea is to divide the \kVDO for a tree rooted by $r$ into multiple sub-problems on subtrees rooted by $r$'s children nodes. For the selection of root $r$, we have two choices of \emph{selecting} $r$ into answers or \emph{not selecting} $r$ into answers. The optimal solution is one of the best summary score among the above two choices. The above step can be repeatedly enumerated for each node as a root in a polynomial time. However, a straightforward implementation of the above dynamic programming algorithm may incur expensive computations. To improve the efficiency, we develop several useful optimization techniques including the using Knapsack dynamic programming techniques to tackle the exponential division enumeration, and reduce the number of all possible states using the closet ancestor. The time complexity of our dynamic algorithm \DP takes $O(nhk^3)$ time in $O(nhk^2)$ space, where $n, h$ are the node size and height in  tree $T$, respectively. We also theoretically analyze the correctness of \DP to achieve optimal answers. 
\revise{To summarize, we compare \greedy and \DP here. On one hand, \greedy finds approximate answers and runs faster than \DP, which is more particularly suitable to give quick summarization answers. 
\greedy supports the zoom-in and zoom-out functions by freely adjusting the parameter $k$ for $k$-sized summarization in real time. 
On the other hand, \DP achieves optimal solutions by taking more cost than \greedy, which is more suitable in those critical applications for quality-priority. Our comprehensive solutions provide the choices to achieve a balanced trade-off between quality and efficiency. }

In addition, we further develop tree reduction techniques to accelerate computations in Section~\ref{sec.vtree}, which is another new technical contribution over \cite{huang2017ontology}. The tree reduction is based on an important observation that a large number of vertices have zero-weights in tree $T$. These vertices with zero-weights may be unimportant for tree summarization, which can be removed from $T$. We then propose a tree reduction method to delete them and shrink the whole tree into a small tree $T^*$, which contains a few nodes with non-zero weights. Our \DP applied on the reduced tree $T^*$ achieves the same optimal solution as the original tree $T$, but runs much faster in practice and also in theoretical time complexity analysis. 
The efficiency and effectiveness of our proposed  
tree reduction algorithm are validated by extensive experiments on real-world datasets. 

To summarize, this paper makes the following contributions:

\squishlisttight
\item We motivate and formulate the problem of tree summarization, which aims at selecting a diverse set of $k$ representative vertices in a weighted tree. 
We identify the desiderata of a good tree summarization, admitting the representativeness, diversity, and high-score coverage simultaneously (Section~\ref{sec.problem}).
\item  We analyze the summary objective function. We  formally prove its monotonicity and submodularity properties, which offer the prospects for developing efficient and approximate algorithms (Section~\ref{sec.analysis}). 
\item We propose an efficient algorithm that can achieve at least $(1-1/e)$ of the optimal in terms of summary objective function. Moreover, we present a graphical visualization method to depict a hierarchical structure based on the obtained summary results, which has been reduced from the original tree to a manageable size. This graphical visualization intends to depict, highlight, and distinguish the important nodes and links within the hierarchal structure, which illustrates a useful application of our tree summarization problem (Section~\ref{sec.greedy}).
\item We develop an exact algorithm \DP based on dynamic programming to achieve optimal solutions for \kVDO. We further propose several optimization strategies for efficient implementation. 
We also analyze the algorithm correctness and complexity of \DP (Section~\ref{sec.dp}).
\item We propose a tree reduction technique to prune zero-weighted vertices in the tree. It can significantly reduce the tree size and generate a small new tree. Based on the newly generated tree, \DP is guaranteed to achieve the same optimal answers in a faster way  (Section~\ref{sec.vtree}).

\item We conduct extensive experiments on five real-world datasets to validate the efficiency and effectiveness of our proposed algorithm. Moreover, we show one case study and one usability evaluation on a real dataset of ACM Computing Classification System, reflecting the practical usefulness of our tree summarization model and algorithms, in terms of graph visualization and users' feedback (Section~\ref{sec.exp}).
\end{list} 

We discuss related work in Section \ref{sec.relate} and conclude the paper in Section~\ref{sec.con}.

%% file: tex/relate.tex
\section{Related Work}\label{sec.relate}

Work closely related to our paper can be categorized into data summarization, 
graph visualization and interactive search, and top-$k$ diversification. 

\stitle{Data summarization.} There exist several studies on data summarization \cite{jing2014complementary, wu2017finding,tian2008efficient,noel2004managing, vcebiric2015query, DBLP:conf/icde/Gou0Z019, kumar2018utility, liu2014distributed, yang2011summary}. \cite{wu2017finding} finds a set of $k$ high-quality and diverse representatives for a surface, which does not consider the ontology structure associated with the data.  
In \cite{vcebiric2015query}, a semi-structured framework is developed to summarize RDF graphs. 
A novel sketch approach is proposed by Gou et al.~\cite{DBLP:conf/icde/Gou0Z019} to summarize the graph streams. It takes linear space and constant update time. Both of these two works design data structures to store and summarize graphs. 
Kumar and Efstathopoulos \cite{kumar2018utility} propose a method of computing utility to summarize and compress graphs. Liu et al. \cite{liu2014distributed} develops several distributed algorithms for graph summarization on the Giraph distributed computing framework. Most of these works use graph compression or subgraph mining to summarize the whole graph structural information. \emph{Different from the above studies, our work considers the problem of data summarization using ontology terminologies, and formulates it as an optimization problem.}
In addition, several works study data summarization on hierarchical data~\cite{agarwal2007efficient, jin2009tree, karloff2011parsimonious, ruhl2018cascading, kimsummarizing, zhu2020hdag, zhu2020top}. 
Agarwal et al. \cite{agarwal2007efficient} propose the parsimonious explanation model to summarize changes in dimension hierarchy. 
Kim et al. \cite{kimsummarizing} propose dynamic programming methods to creates a concise summary of hierarchical multidimensional data. Both \cite{agarwal2007efficient}\cite{kimsummarizing} focus on the changes between two different hierarchies. Recently, Zhu et al.~\cite{zhu2020top} studies a NP-hard problem of top-$k$ graph summarization on DAGs, which is a generalization of our tree summarization problem. \emph{Different from the heuristic graph summarization algorithms~\cite{zhu2020top}, we develop an optimal algorithm for tree summarization using new dynamic programming techniques. }

\stitle{Graph visualization and interactive search.} 
Many works have been carried out on studying graph visualization \cite{jing2011graphical, koutra2015perseus, akoglu2012opavion, DBLP:conf/icde/KrommydaKV19, wu2015efficient, hasani2018tableview, jiang2018vizcs,bhowmick2017graph,yi2017autog}. The problem of graphical visualization using ontology terminologies is investigated to filter the nodes whose aggregate frequencies are less than a given threshold \cite{jing2011graphical}. 
Perseus \cite{koutra2015perseus} is a large-scale graph system developed to enable the comprehensive analysis of large graphs and allow the user to interactively explore node behaviors. OPAvion~\cite{akoglu2012opavion} provide scalable and interactive workflow to accomplish complex graph analysis tasks. Most of these works~\cite{jiang2018vizcs,bhowmick2017graph,yi2017autog} design a graph visualization system to analyze the large scale graphs. 
\emph{Unlike the above graph visualization algorithms and systems, we find $k$ representative vertices to summarize the whole hierarchy.} 
In addition, graph interactive search~\cite{parameswaran2011human, tao2019interactive, zhu2021budget} study a crowdsourcing task to identify the target labels of a given object in a label hierarchy, which allows asking users for a few questions. Recently, Zhu et al.~\cite{zhu2021budget} propose a dynamic programming based algorithm to ask one question with the maximum gain based on $k$ targets. The targets are fixed but unknown in advance. The hierarchy has no vertex weights. \emph{Compare with \cite{zhu2021budget}, although our dynamic programming techniques adopt a similar idea of the Knapsack problem as \cite{zhu2021budget}, we focus on a different problem of tree summarization, which finds a $k$-sized summary vertex set with the largest summary score on a weighted tree where vertices have weights. Moreover, we propose efficient tree reduction techniques especially for tree summarization, which cannot be applied on interactive search problem.}

\stitle{Top-k diversification.} In the literature, a large number of work studies the diversification of top-$k$ query results \cite{qin2012diversifying,ranu2014answering, yang2016diversified, yuan2016diversified, catallo2013top, zhou2010solving, fan2013diversified, li2012measuring}.  A comprehensive survey of top-$k$ query processing can be found in \cite{ilyas2008survey}. A general diversified top-$k$ search problem is defined by Qin et al. \cite{qin2012diversifying}, which only considers the similarity of the search results themselves. In \cite{ranu2014answering}, Ranu et al. propose an index structure NB-Index. It can solve the top-$k$ representative queries on graph databases. \cite{yuan2016diversified} finds top-$k$ maximal cliques which can cover most number of vertices. These works study the top-k diversification on graph databases, subgraph queries, and cliques. The key distinction with these existing studies is that our approach takes a flexible method to 
we investigate a different problem of finding a small set of $k$ nodes to summarize the whole tree with weighted terminologies. 

%% file: tex/problem.tex
\section{problem statement}\label{sec.problem}

In this section, we define basic notions and formalize our problem.

\subsection{Preliminaries}  \label{sec.pre}
We consider a finite set of $n$ elements, $\mathcal{V}$, where the elements with inter-relations are organized into a tree-like structure. 
Let a weighted tree $T=(\mathcal{V}, E, \feq)$ be rooted at $r\in \mathcal{V}$, where 
$E\subseteq \mathcal{V}\times\mathcal{V}$ 
is the edge set and $\feq$ is the node weight function. The tree $T$ contains $n=|\mathcal{V}|$ nodes and $n-1 =|E|$ edges. In addition, the node weight $\feq(v)\in \mathbb{R}^{\geq0}$ is a non-negative real value, which presents the importance of node $v$. We denote an important set of positive nodes as $\mathcal{I}=\{v\in \mathcal{V}: \feq(v)>0\} \subseteq \mathcal{V}$, representing the set of all nodes $v$ with positive weights. For each node $v$ in $T$, we respectively denote the ancestors of node $v$ by $\anc(v)$ and the set of descendants of node $v$ by $\dec(v)$. Note that, we denote that $\anc(v)$ and $\dec(v)$ always contain $v$ throughout this paper, i.e., $v\in \anc(v)$ and $v\in \dec(v)$. 
Furthermore, we denote the children of node $v$ by $\nb^{-}(v) =\{u\in \mathcal{V}: (u, v)\in E, u\notin \anc(v)\}$. A children node $u\in \nb^{-}(v)$ is only one level below the node $v$ in tree $T$.  A node with $|\nb^{-}(v)|=0$ is called a leaf node. 

\begin{definition}[Node Level] Given a tree $T$ rooted at $r$, the level of a tree node $v\in \mathcal{V}$ is the number of hops between $v$ and $r$, denoted by $l(v)$.
\end{definition}

For example, consider a tree $T$ in Figure \ref{fig.example}(a).  
For node $C$, the set of descendants of $C$ is $\dec(C)=\{C, c_0, c_1, c_2, c_3, c_4\}$, and the set of ancestors is $\anc(C)=\{r, C\}$. The level of node $C$ is $l(C)=1$, and the level of node $c_2$ is $l(c_2)=3$. 

\stitle{Desiderata of a good summarization. }
Given a weighted tree $T=(\mathcal{V}, E, \feq)$  and an important set of positive nodes $\mathcal{I} \subseteq \mathcal{V}$, our goal, intuitively, is to select a small set of elements $S$ from $\mathcal{V}$ that depicts a good summarization of the high-score data of $\mathcal{I}$ by satisfying the following three criteria:
\squishnumlist
\item (Diversity) The elements of $S$ should not be very similar;
\item (Small-scale) The size of $S$  is small enough to be easily understood; 
\item (High-score Coverage and Correlation) A summary score function 
$\g(S)$ that measures the coverage and correlation of $S$ on important nodes $\mathcal{I}$,  is high. 
\end{list} 

\subsection{Summary Score Function}
In this subsection, we propose a summary score function \eatTKDE{$\score_{\mathcal{S} }(\mathcal{I} )$} \TKDE{$\g(S)$} by
formalizing the desiderata of diversity, high-score coverage, and correlations in a unified way. We first give the definitions of coverage and correlation below.

\stitle{Coverage.} Given two nodes $x, y$ in tree $T$, we say $x$ covers $y$ if and only if $x$ is one ancestor of $y$, i.e., $y\in \dec(x)$.
In the concept tree $T$, $x$ covers $y$, indicating that $x$ is a more general concept than $y$. This shows that $x$ can be a summary representative of $y$ in a higher level of concept understanding.  For instance,  in Figure \ref{fig.example}(a), node $c_0$ covers a set of nodes \{$c_1, c_2, c_3, c_4$\}, which means $c_0$ can be a good summary of all concepts in \{$c_1, c_2, c_3, c_4$\}.

\stitle{Representative Impact.} Based on the definition of coverage, we define the representative impact as follows.

\begin{definition} [Representative Impact] \label{def.rep}
Given two elements $x, y$ and $y\in \dec(x)$,
we define the representative impact of $x$ on the element $y$ using a function $\rep_{x} :$ 
$$\rep_{x}(y) =  \feq(y)\cdot \dis_x(y)$$,
where $\dis_x: \mathcal{V}  \rightarrow \mathbb{R}^{\ge 0}$ is the summarized relevance function.
\end{definition}

Here, $x$ serves as a candidate representative of $y$. The summarized impact of $x$ on $y$ is proportional to $\feq(y)$, the node weight of $y$, and is discounted by $\dis_x(y)$. Specifically, the summarized relevance of $x$ achieves the maximum at $y=x$, and decreases for $y$ further away from $x$. Note that, if $x$ does not cover $y$, i.e., $y\notin \dec(x)$, then $\dis_x(y) =0 $ and certainly $\rep_{x}(y) = 0$. 
In this paper, we suggest one natural choice of correlation function 
\begin{equation}\label{eq:nm0}
\dis_x(y) =
\left\{
\begin{aligned}
\frac{1}{l(y)-l(x)+1}, \  \text{if }~ y\in \dec(x) \\
0,~~~~~~~~~~~~~~~~~~~~~~~\text{otherwise }
\end{aligned}
\right.
\end{equation}
For example, consider the tree $T$  and the  weight function of elements as $\feq(\cdot)$ in Figure \ref{fig.example}(a). For nodes $B$ and $b_1$ with the level $l(B)=1$ and $l(b_1)=2$,  the summarized relevance of $B$ on $b_1$ is $\dis_B(b_1) = 1/2$, and thus representative impact of $B$ on $b_1$ is $\rep_{B}(b_1) =  \feq(b_1)\cdot \dis_B(b_1) = 30\times 1/2= 15$. On the other hand, the summarized relevance of $r$ on $b_1$ is $\dis_{r}(b_1) = 1/3$, and the representative impact $\rep_{r}(b_1)=10 < \rep_{B}(b_1)$, indicating that $B$ is a better summarized representative outperforming $r$, due to the more specification of $B$ compared to $r$. Our models can adopt other settings of $\dis_x(y)$ satisfying the principle of summarized relevance, and also our proposed techniques can be easily extended to solve a variant of problems with different  $\dis_x(y)$ functions. 

\stitle{Summary score.}  Given a set $S \subseteq \mathcal{V}$ of representative elements, we define the summary score of $S$ on an input element $y\in \mathcal{V}$, denoted by $\smy_{S}(y)$, as the maximum impact $y$ among all individual representatives:
\begin{equation}
\smy_{S}(y) = \max_{x\in S \cap \anc(y)} \rep_{x}(y).
\end{equation}

Intuitively, each input element $y$ is to be represented by some ancestor of $y$ that appears in $S$ (a.k.a. $x\in S\cap \anc(y)$) and has the maximum summary impact on $y$. Based on the definition of summary score, the total summary impact of $\mathcal{S}$ on all elements of $\mathcal{I}$ is defined as:
\begin{equation}
\g(S)= \sum_{y\in \mathcal{I}} \smy_{S}(y) = \sum_{y\in \mathcal{I}}\max_{x\in S\cap\anc(y)} (\feq(y)\cdot \dis_x(y)).
\end{equation}

To recap, the problem of \underline{t}ree \underline{s}ummarization with \underline{w}eighted terminologies (\smyp) studied in this paper can be formally formulated as follows.

\stitle{kWTS-problem}. Given a weighted tree $T=(\mathcal{V}, E, \feq)$, an important set of positive nodes $\mathcal{I} \subseteq \mathcal{V}$, and an integer $k\in \mathbb{Z}^{+}$, the problem is to find a set of representative nodes $S \subseteq \mathcal{V}$, such that $S$ achieves the maximum score $\g(S)$ with  $|S|= k$.

\begin{example}
We use the example in Figure \ref{fig.example} to illustrate our \smyp and set $k=5$. To summarize the tree with important set $\mathcal{I} =\{A, a_1, a_2, a_3, b_1, r, c_0, c_1, c_2, c_3, c_4\}$ in Figure \ref{fig.example}(a), an optimal solution is the summary graph $S=\{r, A, a_1, b_1, c_0\}$ in  Figure \ref{fig.example}(d). For node $a_1\in \mathcal{I}$, the best representative of $S$ is $a_1$ and the summary score of $S$ on $a_1$ is $\smy_{S}(a_1)=40\times1=40$. Overall, the total summary score of $S$ is $\g(S)=\sum_{x \in \mathcal{I}} \smy_S(x) = \rep_A(A) + \rep_{a_1}(a_1) + \rep_{A}(a_2) + \rep_{A}(a_3) + \rep_{b_1}(b_1) + \rep_{r}(r) + \rep_{c_0}(c_0) + \rep_{c_0}(c_1) + \rep_{c_0}(c_2) + \rep_{c_0}(c_3) + \rep_{c_0}(c_4) = 30 + 40 + 10 + 10 + 30 + 10 + 10 + 5 + 5 + 5 + 5 = 160$. 
\end{example}

\begin{table}[t]
\centering
\caption{Frequently used notations.}\label{tab.notat}
\vspace{-0.3cm}
\scalebox{1.0}{
\begin{tabular}{|l|l|}
\toprule
Notation &  Description \\
\midrule
$\feq(v)$& the importance of vertex v\\
$\mathcal{I}$& the set of vertices with $\feq(v) > 0$\\
$\anc(v)$/$\des(v)$& the set of ancestors/descendants of vertex v\\
$\nb^{-}(v)$& the set of children of vertex v\\
$\l(v)$& the level of vertex v\\
$\dis_u(v)$& the correlation impact of u on v\\
$\rep_u(v)$& the representative score of u on v\\
$\g(S)$& the summary score of S for all vertices\\
$\smy_S(v)$& the summary score of S on the vertex v\\
$\triangle_{g}(x|S)$& $\triangle_{g}(x|S)= \g(S\cup\{x\}) - \g(S)$\\
$T_u$& the subtree rooted with $u$\\
$S_u^k$& the summary set of selecting $k$ vertices in $T_u$\\
$\DP(u, k, S)$& the largest summary score $\g(S_u^k\cup S)$ in $T_u$\\
$\mathcal{Y}(u ,k, S)$/$\mathcal{N}(u ,k, S)$& $\DP(u, k, S)$ with/without selecting $u$ in $S_u^k$\\
\bottomrule
\end{tabular}
}
\vspace{-0.3cm}
\end{table}

%% file: tex/analysis.tex
\section{Problem Analysis}\label{sec.analysis}

In this section, we analyze the properties of the objective score function of our problem.

\stitle{Monotonity and Submodularity} \label{sec.modular}
A set function $f: 2^U \rightarrow \mathbb{R}^{\ge 0}$ is said to be submodular provided for all sets $S\subset T\subset U$ and element $x\in U\setminus T$, $f(T\cup\{x\}) - f(T) \le f(S\cup\{x\}) - f(S)$, i.e., the marginal gain of an element has the so-called ``diminishing returns'' property.

\begin{lemma}
$\g$ is monotone, i.e., for all $S_1, S_2 \subseteq \mathcal{V}$  such that $S_1 \subseteq S_2$, we have $\g(S_1) \leq \g(S_2)$.
\end{lemma}
\begin{proof} 
Since $S_1 \subseteq S_2$, for any element $y\in \mathcal{I}$, $\max_{x\in S_2} \dis_x(y) \geq \max_{x\in S_1} \dis_x(y)$, which is trival. Now, we have
$\g(S_2) - \g(S_1) = \sum_{y\in \mathcal{I}}(\max_{x\in S_2} (\feq(y)\cdot \dis_x(y))) - \sum_{y\in \mathcal{I}}(\max_{x\in S_1} (\feq(y)\cdot \dis_x(y))) = \sum_{y\in \mathcal{I}} \feq(y) \cdot (\max_{x\in S_2} \dis_x(y) - \max_{x\in S_1} \dis_x(y) ))  \geq 0$. As a result, $\g(S_1) \leq \g(S_2)$ holds.
\end{proof}

Given a summary node $x\in S$, let the set of nodes that take $x$ as their summary node, denoted by $\Upphi_{S}(x) = \{y\in \dec(x): \smy_{S}(y) = \rep_{x}(y)\}$.

\begin{lemma}\label{lemma.submodular}
$\g$ is submodular.
\end{lemma}

\begin{proof}
Give two sets $S\subset T\subset \mathcal{V}$  and an element $x\in \mathcal{V}\setminus T$, let $T'= T\cup\{x\}$ and $S'= S\cup \{x\}$. We establish the correctness of Lemma \ref{lemma.submodular} by following three facts below. 

First, for any element $y\in \mathcal{V}$, $\smy_{T}(y) \geq \smy_{S}(y)$ and $\smy_{T'}(y) \geq \smy_{S'}(y)$ holds. Second, $\Upphi_{T'}(x)$ $\subseteq \Upphi_{S'}(x)$. Since $\forall y \in \Upphi_{T'}(x)$, we have $\rep_{x}(y) = \smy_{T'}(y) \geq  \smy_{S'}(y)$ and $\rep_{x}(y) \leq \smy_{S'}(y)$ for $x\in S'$. As a result, we obtain $\rep_{x}(y) = \smy_{S'}(y)$ and $y\in \Upphi_{S'}(x)$. Therefore, $\Upphi_{T'}(x)$ $\subseteq \Upphi_{S'}(x)$ holds. Third, we have $\g(T') - \g(T) =\sum_{y\in \mathcal{V}} (\smy_{T'}(y)- \smy_{T}(y))$ $= \sum_{y\in \Upphi_{T'}(x)} (\rep_{x}(y) -\smy_{T}(y))$. Thus, we can obtain $\g(S') - \g(S)$ 
$ = \sum_{y\in  \Upphi_{S'}(x)}$ $(\rep_{x}(y) -\smy_{S}(y))$ 
$\geq $ $\sum_{y\in \Upphi_{T'}(x))} $ $(\rep_{x}(y) -\smy_{S}(y))$
$\geq $ $\sum_{y\in \Upphi_{T'}(x))} $ $(\rep_{x}(y) -\smy_{T}(y))$
$=\g(T') - \g(T)$. As a result, $\g(S') - \g(S) \geq \g(T') - \g(T)$. 
\end{proof}

In view of the fact that $\g$ is monotone and submodular, we infer that the prospects for developing an efficient approximation algorithm using greedy strategies are promising.

%% file: tex/greedy.tex
\section{GTS Algorithm}\label{sec.greedy}

\begin{algorithm}[t]
\caption{\greedy($T$, $\mathcal{I}$, $k$)} \label{algo:greedy}
\begin{algorithmic}[1]
\Require A tree $T=(\mathcal{V}, E, \feq)$, an important node set $\mathcal{I} \subseteq \mathcal{V}$
, a number $k$.
\Ensure A set of $k$ summary elements $S$.

\State Let $S\leftarrow \emptyset$; 
\While{$|S|< k$}
	\State $x^{*} \leftarrow \arg\max_{x\in \mathcal{V}/S} \triangle_{g}(x|S)$;
	\State $S \leftarrow S \cup \{x^*\}$;
\EndWhile
\State \textbf{return} $S$;
\end{algorithmic}
\end{algorithm}

In this section, we present a greedy algorithm that can produce a solution  achieving at least  $(1- 1/e) \approx 62\%$ of the optimal score $g(S^*)$. In the following, we first give the framework of our greedy algorithm called \greedy. Then, we show its approximation guarantee and present several techniques for improving its efficiency. Finally, we discuss how to use the answer of selected vertices by \greedy to represent the whole tree $T$. 

\subsection{A Greedy Algorithm GTS}

\stitle{Marginal gain.} We begin with marginal gain. Monotonicity of function $\g$ implies that for any $S \subseteq \mathcal{V}$  and $x \in \mathcal{V}$, we have $\triangle_{g}(x|S)= g(S\cup\{x\}) - g(S) \geq 0$. The term $\triangle_{g}(x|S)$ is called the marginal gain 
of $x$ to the set $S$. We would like to add the node with the largest marginal gain into the answer. This greedy strategy motivates the following algorithm \greedy.

\stitle{Algorithm overview.}  \greedy starts out with an empty solution set $S=\emptyset$. In each subsequent iteration, \greedy iteratively adds one more summary node $x^*$ to solution $S$, which grows the answer set by one. This summary node $x^*$ is chosen
from the remaining candidate elements $\mathcal{V}/S$ such that it achieves the largest marginal gain, i.e., $x^{*} \leftarrow \arg\max_{x\in \mathcal{V}/S} \triangle_{g}(x|S)$. Finally, \greedy returns $S$ after $|S|= k$. The detailed description is presented in Algorithm \ref{algo:greedy}.

\stitle{Computing $\triangle_{g}(x|S)$.} We present an efficient algorithm (Algorithm \ref{algo:gain}) for computing the marginal gain $\triangle_{g}(x|S)$. 
Let $S'=S\cup \{x\}$, and $T_x$ be a subtree of $T$ rooted at $x$ (lines 1-2). The procedure computes $\Upphi_{S'}(x)$ by performing one traversal of tree $T_x$ and finding all nodes regarding $x$ as its new summary node. Afterwards, if we can find the nearest ancestor $z$ of $x$ in $S$, i.e. $\anc(x) \cap S \neq \emptyset$, and calculate the marginal gain $\triangle_{g}(x|S) =\sum_{y\in \Upphi_{S'}(x)} (\rep_{x}(y)-\rep_{z}(y))$; otherwise, if such an ancestor $z$ does not exist, the algorithm directly returns $\triangle_{g}(x|S) =\sum_{y\in \Upphi_{S'}(x)} \rep_{x}(y)$.

\begin{algorithm}[t]

\caption{Computing $\triangle_{g}(x|S)$} \label{algo:gain}
\begin{algorithmic}[1]
\Require A tree $T$, an important node set $\mathcal{I}$, a summary set $S$, a node $x\in\mathcal{V}$.
\Ensure $\triangle_{g}(x|S)$.

\State $S'\leftarrow S\cup \{x\}$;
\State Compute $\Upphi_{S'}(x)= \{y\in \dec(x): \smy_{S'}(y) = \rep_{x}(y)\}$;
\If {\anc(x) $\cap S \neq \emptyset$}
	\State Let $z\in S$ be the nearest ancestor of $x$;
	\State $\triangle_{g}(x|S) =\sum_{y\in \Upphi_{S'}(x)} (\rep_{x}(y)-\rep_{z}(y))$;
\Else
	\State $\triangle_{g}(x|S) =\sum_{y\in \Upphi_{S'}(x)} \rep_{x}(y)$;
\EndIf
\State \textbf{return}  $\triangle_{g}(x|S)$;
\end{algorithmic}
\end{algorithm}

\stitle{Approximation analysis. } \cite{nemhauser1978analysis} shows that a greedy algorithm provides a $(1-1/e)-$approximation for maximizing a monotone submodular set function with cardinality constraint. 
Our method \greedy is one instantiation of this algorithm for \smyp.

\begin{theorem}
Let $S$ be the answer obtained by \greedy, and $S^*$ be the optimal answer,
$\g(S) \geq (1-\frac{1}{e})\cdot \g(S^*)$ holds.
\end{theorem}

\stitle{Complexity analysis.} 
\eatTKDE{Assume that the height of tree $T$ is $h$, and a subtree of $T$ rooted at $x\in\mathcal{V}$ as $T_x$, where $|T_x| \leq |\mathcal{V}|=n$. The computation of $\triangle_{g}(x|S)$ consumes $O(|T_x|)$ time of the traversal of $T_x$ and $O(h)$ time to find the ancestors of $x$, which takes $O(|T_x|+h) \subseteq O(n)$ time in total. In Algorithm \ref{algo:greedy}, at each iteration of selecting the node with the maximum marginal gain, it computes $\triangle_{g}(x|S)$ at most $n$ times. As a result, the overall time complexity of Algorithm \ref{algo:greedy} is $O(n^2k)$ time in worst cases. The space complexity is $O(n)$.
}
Assume that the height of tree $T$ is $h$. A subtree of $T$ rooted at $x\in\mathcal{V}$ is denoted as $T_x$. The computation of marginal gain $\triangle_{g}(x|S)$ first takes $O(|T_x|)$ time for the subtree traversal of $T_x$. Then, it takes $O(\l(x))$ time to find the ancestors of $x$. Hence, the computation of $\triangle_{g}(x|S)$ takes $O(|T_x|+\l(x))$ time in total. At each iteration, Algorithm \ref{algo:greedy} selects a node with the maximum marginal gain, which needs to compute $\triangle_{g}(x|S)$ for all nodes $x$ in worst. To select a summary vertex, it costs $O(\sum_{x \in \mathcal{V}}{|T_x|}+\sum_{x \in \mathcal{V}}{\l(x)}) = O(2\sum_{x \in \mathcal{V}}{\l(x)}) \subseteq O(nh)$. As a result, the time complexity of Algorithm \ref{algo:greedy} is $O(nhk)$ time. The space complexity of Algorithm~\ref{algo:greedy} is $O(n)$.

\subsection{Graph Visualization based on Summary Answers} \label{sec.visualization}

In this section, we discuss how to use the obtained answer $S$ to summarize the whole tree $T$ in graph visualization. 
Based on the obtained $k$ representative vertices, it is ready to generate a small summary tree for graph visualization. We first create a virtual root $r$. We then start from each vertex $v\in S$ and add an edge path between $v$ to the lowest ancestor in $S$; If such an ancestor does not exist, we add an edge path between $v$ and the virtual root.

\begin{table}[t]
\centering
\caption{The running steps of \greedy 
applied on tree $T$ in Figure~\ref{fig.example}(a). It shows the marginal gains $\triangle_{g}(x|S)$ for partial vertices in $T$.}\label{table.greedy}
\vspace{-0.3cm}
\scalebox{1.0}{
\begin{tabular}{c|ccccccccc}
\toprule
Step& $r$& $A$& $B$& $C$& $a_1$& $a_2$& $a_3$& $b_1$& $c_0$ \\
\midrule
Step1& 65& {\textcolor{red}{\textbf{70}}}& 15& 18.3& 40& 20& 20& 30& 30\\
Step2& {\textcolor{red}{\textbf{33.3}}}& /& 15& 18.3& 20& 10& 10& 30& 30\\
Step3& /& /& 5& 5& {\textcolor{red}{\textbf{20}}}& 10& 10& 20& 16.6\\
Step4& /& /& 5& 5& /& 10& 10& {\textcolor{red}{\textbf{20}}}& 16.6\\
Step5& /& /& 5& 5& /& 10& 10& /& {\textcolor{red}{\textbf{16.6}}}\\
\bottomrule
\end{tabular}
}
\vspace{-0.5cm}
\end{table}

\begin{example}
We use the tree $T$ in Figure~\ref{fig.example}(a) to illustrate the running steps of \greedy algorithm and graph visualization using summary answers. Suppose that $k=5$. We apply Algorithm~\ref{algo:greedy} on $T$. 
Table~\ref{table.greedy} shows the marginal gains $\triangle_{g}(x|S)$ of vertices $x\in V$ without $c_1, c_2, c_3$ and $c_4$. At the first step, we calculate $\triangle_{g}(x|S)$ of all vertices in $T$ and then choose the vertex $A$ with the largest marginal gain $\triangle_{g}(A|S) = 70$ for $S=\emptyset$. Next, we update the $\triangle_{g}(x|S)$ for remaining vertices and choose the vertex $r$ with the largest marginal gain $\triangle_{g}(r|S) = 33.3$ for $S=\{A\}$ at the second step. Similarly, we select other three vertices $a_1, b_1, c_0$ as answers and the summary set is $S = \{A, r, a_1, b_1, c_0\}$. Finally, we use five representative vertices $S$ to depict the summarized graph visualization as shown in Figure~\ref{fig.example}(d). We connect vertex $a_1$ to vertex $A$ by adding an edge $(A, a_1)$ as $A$ is the lowest ancestor of $a_1$ in $S$. Similarly, we connect three vertices $A, b_1, c_1$ to root $r$ by adding edges and skip their connections to vertices $B, C$ that do not belong to the answer $S$. 
\end{example}

%% file: tex/dp.tex
\section{OTS algorithm} \label{sec.dp}

In this section, we introduce an exact algorithm for tree summarization, which finds an optimal answer $S^*$ of $k$ representative nodes to achieve the maximum summary impact, i.e., $S^* = \arg\max_{S\subseteq V, |S|=k} \g(S)$.  
To this end, we propose a dynamic programming algorithm \DP and develop several optimization techniques to accelerate the search process. Furthermore, we also analyze the complexity and correctness of \DP. 

\subsection{Solution Overview}

We first use a toy example to show the limitation of greedy algorithm \greedy, which is still far from the optimal answer. Consider the tree $T$ in Figure~\ref{fig.dp.example} and assume $k = 2$. The \greedy algorithm firstly selects the vertex $v_2$ with the maximum margin of $43$ and secondly selects the vertex $v_4$. \greedy generates the answer $S = \{v_2, v_4\}$, which achieves the summary score $\g(S) = 64$. However, this answer $\g(S) = 64$ is not optimal. The best answer is $S^* = \{v_3, v_4\}$ with the summary score $\g(S^*) = 81$. Actually, after selecting the vertex $v_4$, the marginal gain of $v_2$ is reduced and the vertex $v_3$ becomes a better selection. To dismiss the limitation of local optimality by greedy algorithm, we consider an alternative method of dynamic algorithm in terms of global optimality.

Figure~\ref{fig.dp.overview} shows an overview framework of our dynamic programming algorithm \DP. As shown in the above example, although \greedy is an efficient and approximate algorithm, it cannot find the optimal solution in some cases due to its local optimality. To find the global optimality, \DP considers a subproblem of finding $k^\prime \leq k$ representative nodes optimally in a subtree $T_u$ rooted by a vertex $u \in \des(r)$ as shown in Figure~\ref{fig.dp.overview}. Moreover, we consider to have an existing partial answer $S$ already and combine $S$ with another set $S_{u}^{k^\prime}$ of $k^\prime$ nodes to form a globally optimal answer, i.e., $|S\cup S_{u}^{k^\prime}| =k$. Obviously, let $S=\emptyset$, $u=r$, and $k^\prime =k$, thus this subproblem is the same as the original \kVDO. Thus, the problem is how to find additional $k^\prime$ optimal vertices in the subtree with selected set $S$. We consider two cases of whether we select vertex $u$ or not. On one hand, if we select vertex $u$ into the answer $S$, for each children node $v_1, v_2 ..., v_x$, the sub-problem is how to find additional $k_x$ optimal vertices in the subtrees rooted by $v_x$ with an existing answer $S \cup \{u\}$ and $\sum{k_x} \leq k^\prime - 1$; On the other hand, if we do not select vertex $u$ into the answer $S$, for each children node $v_1, v_2 ..., v_x$, the sub-problem is how to find additional $k_x$ optimal vertices in $T_{v_x}$ with an existing answer $S$ and $\sum{k_x} \leq k^\prime$. The optimal answer is the best solution among the above two answers. 

\begin{figure}[t]
\centering{
\subfigure[Tree $T$]{
\label{example.dp.1}
\includegraphics[width=0.28\linewidth]{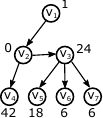} }
\quad
\subfigure[\greedy answer]{
\label{example.greedy}
\includegraphics[width=0.28\linewidth]{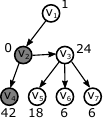} }
\quad
\subfigure[\DP answer]{
\label{example.dp.2}
\includegraphics[width=0.28\linewidth]{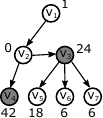} }
}
\caption{A motivation example of comparing different answers between \greedy and \DP. For the same tree $T$ in Figure~\ref{example.dp.1}, \greedy generates an answer $S = \{v_2, v_4\}$ with $\g(S) = 64$ in Figure~\ref{example.greedy}, which is worse than the optimal answer $S^* = \{v_3, v_4\}$ with $\g(S^*) = 81$ in Figure~\ref{example.dp.2}. }
\label{fig.dp.example}
\end{figure}

\begin{figure}[t]
\centering{
{
\label{dp.overview}
\includegraphics[width=0.4\linewidth]{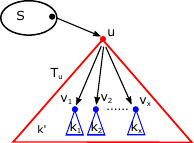} }
}
\caption{A solution overview of \DP algorithm.}
\label{fig.dp.overview}
\vspace{-0.3cm}
\end{figure}

\subsection{Dynamic Programming Algorithm}

In the following, we give the detailed formulations of states, sub-problems and the algorithm.

\stitle{States}. We begin with a definition of state $\DP(u, k, S)$ in dynamic programming. Given a tree $T$, a vertex $u\in V$, a number $k$, and a set of summary vertices $S \subseteq V$, $\DP(u, k, S)$ represents an optimal solution of the  \kVDO in a tree $T_u$.  
That is selecting the additional $k$ summary vertices $S_u^k$ from $T_u$ into $S$ to achieve the largest summary score $\g(S_u^k\cup S)$ in the tree $T_u$. Note that $S_u^k \subseteq \des(u)$ and $S \cap \des(u) = \emptyset$. 
An optimal answer of the \kVDO in $T$ is $\DP(r, k, \emptyset)$, where $r$ is the root of $T$.

\stitle{Divide a state into sub-problems}. For the state $\DP(u, k, S)$, we divide it into two sub-problems. For the root vertex $u$, we have the choice of two cases: \emph{Yes-case} and \emph{No-case}. Generally, the choice of \emph{Yes-case} is selecting $u$ into the existing answer $S$, denoted as $\mathcal{Y}(u,k, S)$; the other choice of \emph{No-case} is not selecting $u$ into $S$, denoted as  $\mathcal{N}(u,k, S)$. Intuitively, the best answer of $\DP(u, k, S)$ should be one between \emph{Yes-case} and \emph{No-case}, i.e., 
\begin{equation}\label{eq.max}
\begin{aligned}
\DP(u, k, S) = \max\{\mathcal{Y}(u,k, S), \mathcal{N}(u,k,S)\}, 
\end{aligned}
\end{equation}
where $\mathcal{Y}(u,k, S)$ and $\mathcal{N}(u,k, S)$ are respectively shown in Eq.~\ref{eq.Yes} and Eq.~\ref{eq.No}. 

For $\mathcal{Y}(u,k, S)$, it adds $u$ into $S$ and has a new summary set $S\cup\{u\}$. Thus, the summary score for vertex $u$ is obviously $\feq(u)$, as shown in the first term of Eq.~\ref{eq.Yes}. In addition, the number of candidate representative vertices decreases by one, i.e., $k-1$. The optimal solution of $\mathcal{Y}(u,k, S)$ needs to explore all possible assignments of $k-1$ representative vertices into the trees rooted by $u's$ out-neighbors (a.k.a. children), as shown in the second term of Eq.~\ref{eq.Yes}. Specifically, we have

\begin{equation}\label{eq.Yes}
\begin{aligned}
\mathcal{Y}(u,k, S) &=\feq(u)+ \max\{\sum_{x\in N^-(u)} \DP(x, k_x, S\cup\{u\})\}    \\ 
 &  
 \text{ subject to }  \sum_{x\in N^-(u)} k_x= k-1.\\
\end{aligned}
\end{equation}

 For $\mathcal{N}(u,k, S)$, it does not choose $u$ and has an unchanged summary set $S$. Thus, the summary score for vertex $u$ by $S$ is calculated as $\smy_S(u)$, as shown in the first term of Eq.~\ref{eq.No}. The number of candidate representative vertices is still $k$. The optimal solution of $\mathcal{N}(u,k, S)$ needs to explore all possible assignments of $k$ representative vertices into the trees rooted by $x\in N^-(u)$,  as shown in the second term of Eq.~\ref{eq.No}. Specifically, we have

\begin{equation}\label{eq.No}
\begin{aligned}
\mathcal{N}(u,k, S) &= \smy_S(u) + \max\{\sum_{x\in N^-(u)} \DP(x, k_x, S)\}  \\ 
 &  
 \text{ subject to }  \sum_{x\in N^-(u)} k_x= k.\\
\end{aligned}
\end{equation}

\stitle{OTS algorithm.} Algorithm~\ref{algo:dp} implements a dynamic programming algorithm for the \kVDO in a weighted tree $T$. The algorithm computes an optimal summary score $\g(u, k, S)$ for the state $\DP(u,$ $ k, S)$, which is recorded to conveniently use and avoid recomputing. If $\g(u, k, S)$ has been computed before, the score $\g(u, k, S)$ can be directly returned (line 8); Otherwise, it computes the state $\DP(u, k, S)$ dynamically (lines 2-7). The algorithm first checks the number $k$. If $k\geq 1$, it explores to select $k$ representative vertices in subtree $T_u$ via Eq.~\ref{eq.max} (line 3), by invoking two procedures of $\mathcal{Y}(u,k, S)$ in Eq.~\ref{eq.Yes} (lines 9-12) and  $\mathcal{N}(u,k, S)$ in Eq.~\ref{eq.No} (lines 13-16); Otherwise, for $k=0$, Algorithm~\ref{algo:dp} then computes the summary score equals $\smy_S(u) + \sum_{x \in \nb^{-}(u)}\DP(x, k, S)$, which is the representative score $\smy(S, u)$ add the sum of the summary score in trees rooted by $u$'s children (lines 5-7). Computing $\DP(r, k, \emptyset)$ by Algorithm~\ref{algo:dp}  produces an optimal answer $\g(r, k, \emptyset)$ for $T$. Note that $S_u^k$ is the union of all the selection sets $S_x^{k_x}$ for $x\in N^-(u)$.

\begin{algorithm}[t]
  \caption{\DP$(u, k, S)$} \label{algo:dp}
  \begin{algorithmic}[1]
    \Require A tree $T=(\mathcal{V}, E, \feq)$, an important node set $\mathcal{I}$, a vertex $u\in \mathcal{V}$, a number $k$, a set of summary vertices $S$.
    \Ensure An optimal summary score $\g(u, k, S)$.

    \If{$\g(u, k, S)$ has not been computed}
    \If{$k\geq 1$}
    \State $\g(u, k, S) \leftarrow \max\{\mathcal{Y}(u,k, S), \mathcal{N}(u,k, S)\}$;
    \Else
    \State $\g(u, k, S) \leftarrow \smy_S(u)$;
    \For{vertex $x \in \nb^{-}(u)$}
    \State $\g(u, k, S) \leftarrow \g(u, k, S)+ \DP(x, k, S)$;
    \EndFor
    \EndIf
    \EndIf
    \State \Return{$\g(u, k, S)$};
\vspace{-0.3cm}
    \Statex
    \Procedure {$\mathcal{Y}$}{$u,k, S$}
    \Statex //\emph{Yes-case: the answer $S_{\mathcal{Y}}$ contains  $u$.}
    \State $k_{\mathcal{Y}}\leftarrow k-1$; $S_{\mathcal{Y}} \leftarrow S\cup\{u\}$;
    \State Enumerate the assignment of $k_x$ for all vertices $x\in N^-(u)$  such that $\sum_{x\in N^-(u)} k_x= k_{\mathcal{Y}}$ to achieve the following optimization via Eq.~\ref{eq.Yes}:
    \Statex \text{\ \ \ \ \ } $OPT_\mathcal{Y} \leftarrow \max\{\sum_{x\in N^-(u)} \DP(x, k_x, S_{\mathcal{Y}})\}$;
    \State \Return{$\feq(u)+ OPT_\mathcal{Y}$};   
    \EndProcedure
\vspace{-0.3cm}
    \Statex
    \Procedure {$\mathcal{N}$}{$u,k, S$}
    \Statex //\emph{No-case: the answer $S_{\mathcal{N}}$ contains no $u$.}
    \State $k_{\mathcal{N}}\leftarrow k$; $S_{\mathcal{N}} \leftarrow S$;
    \State Enumerate the assignment of $k_x$ for all vertices $x\in N^-(u)$  such that $\sum_{x\in N^-(u)} k_x= k_{\mathcal{N}}$ to achieve the following optimization via Eq.~\ref{eq.No}:
    \Statex \text{\ \ \ \ \ } $OPT_\mathcal{N} \leftarrow \max\{\sum_{x\in N^-(u)} \DP(x, k_x, S_{\mathcal{N}})\}$;
    \State \Return{$\smy_S(u)+ OPT_\mathcal{N}$};  
    \EndProcedure
  \end{algorithmic}
\end{algorithm}

\begin{table}[t]
\centering
\vspace{-0.3cm}
\caption{The DP states of  $\DP(u, k, S)$ in Algorithm~\ref{algo:dp}.}\label{table:dp}
\vspace{-0.3cm}
\scalebox{0.9}{
\begin{tabular}{ccc|ccc|c}
\toprule
$u$& $k$& $S$& $\mathcal{Y}(u, k, S)$& $\mathcal{N}(u, k, S)$& \DP(u, k, S)& $S_u^k$ \\
\midrule
$v_7$& 0& $\{v_3\}$& /& {\textcolor{red}{\textbf{3}}}& 3& $\emptyset$\\
$v_6$& 0& $\{v_3\}$& /& {\textcolor{red}{\textbf{3}}}& 3& $\emptyset$\\
$v_5$& 0& $\{v_3\}$& /& {\textcolor{red}{\textbf{9}}}& 9& $\emptyset$\\
$v_4$& 1& $\emptyset$& {\textcolor{red}{\textbf{42}}}& 0& 42& $\{v_4\}$\\
$v_3$& 1& $\emptyset$& {\textcolor{red}{\textbf{39}}}& 9& 39& $\{v_3\}$\\
$v_2$& 2& $\emptyset$& 64& {\textcolor{red}{\textbf{81}}}& 81& $\{v_3, v_4\}$\\
$v_1$& 2& $\emptyset$& 57.5& {\textcolor{red}{\textbf{81}}}& 81& $\{v_3, v_4\}$\\
\bottomrule
\end{tabular}
}
\vspace{-0.5cm}
\end{table}

\begin{example}
Figure~\ref{example.dp.2} shows an example of applying \DP algorithm with $k = 2$ on $T$ in Figure~\ref{example.dp.1}. Table \ref{table:dp} shows the value and the selection set of some \DP state. The max value of $\DP(v_1, 2, \emptyset) = \DP(v_2, 2, \emptyset) = \DP(v_3, 1, \emptyset) + \DP(v_4, 1, \emptyset) = (\feq(v_3) + \DP(v_5, 0, \{v_3\})$ $ +$ $ \DP$ $(v_6, 0, \{v_3\}) + \DP(v_7, 0, \{v_3\})) + 42 = 24 + 9 + 3 + 3 + 42 = 81$. The selection set $S_1^2 = S_2^2 = S_3^1 \cup S_4^1 = \{v_3\} \cup \{v_4\} = \{v_3, v_4\}$.
\end{example}

\subsection{Implementing Optimizations}
In this section, we propose several useful optimizations to improve the efficiency of Algorithm~\ref{algo:dp}. This is because a straightforward implementation of Algorithm~\ref{algo:dp} takes $O(\sum_{v\in V} k\cdot \#\text{k-assign} \cdot \#S)$ $\subseteq$ $O(\sum_{v\in V}   k\cdot k^{|N^-(u)|} \cdot \binom{n}{k})$ time. For procedures $\mathcal{Y}(u,k, S)$  and  $\mathcal{N}(u,k, S)$, it takes $O(\#\text{k-assign})= O(k^{|N^-(u)|})$ time to enumerate the choices of dividing $k$ values into $|N^-(u)|$ buckets. Moreover,  for the enumeration of all possible answers $S$, it takes $O(\binom{n}{k})$ time. In the following, we optimize the $\#\text{k-assign}$ and $\#S$.  

\stitle{Reduce $\#\text{k-assign}$ by Knapsack Dynamic Programming}.
We propose to use Knapsack dynamic programming techniques \cite{pisinger1995algorithms} to tackle the exponential enumeration in division. We reformulate the enumeration problem in procedure $\mathcal{Y}(u, k, S)$ (line 11 of Algorithm~\ref{algo:dp}) and  $\mathcal{N}(u,k, S)$ (line 15 of Algorithm~\ref{algo:dp}) as the Knapsack problem.  Assume that a number $k$ represent the total capacity. Given a set of vertices $N^-(u)=\{x_1, ..., x_l\}$, for each vertex $x_i$ where $1\leq i\leq l$,  $\DP(x_i, k_{x_i}, S)$ represents an item having an item value of $\g(x_i, k_{x_i}, S)$ and an item volume of $k_{x_i}\leq k$. We assume that $F(i, k')$ is the state that the max value of the first $i$ items with a total of $k'$ capacity. The equation of state transformation is shown as follows.
$$F(i, k') = \max_{0\leq j \leq k'}(F(i - 1, k' - j) + \DP(x_i, j, S)).$$
For initialization, we set $F(i, 0) =0$ for $1\leq i\leq l$. Moreover, $F(l, k)= \max \{\sum_{x\in N^-(u)} \DP(x, k_x, S)\}$ with the constraint $\sum_{x\in N^-(u)} k_x= k$, is the largest summary score for a subtree rooted by $u$ with parameters $S$ and $k$. Hence, we can just enumerate each node $x_i$ in the set $N^-(u)$ for $1\leq i\leq l$ and $0\leq k'\leq k$ to find the maximum summary impact value. This method of dynamic programming can reduce the time complexity from $O(k\cdot k^{|N^-(u)|})$  to $O(|N^-(u)|k^2)$. 

\stitle{Reduce the number of states}. We reduce the number of all possible answers $S$ in $\DP(u,k, S)$ from $O(\binom{n}{k})$ to $O(h)$, where $h$ is the height of $T$. Given a summary set $S$ and a tree $T_u$ rooted by $u$, for each vertex $v\in T_u$, the score of $\smy(S, v)$ only depends on the nearest ancestor of $u$ in $S$, denoted as $\na_u(S)=\arg\min\{\dist \langle v, u \rangle: v\in \anc(u)\cap S\}$.  There exist at most $|\anc(u)|$ different ancestors. 
Instead of $\DP(u,k, S)$, we reformulate the state as $\DP(u, k, \na_u(S))$. This reduces $\#S$ from $O(\binom{n}{k})$ to $O(h)$. Thus, the total number of $\DP(u, k, \na_u(S))$ states is $O(nkh)$. 

\subsection{Correctness and Complexity}

In this section, we prove the correctness of Algorithm~\ref{algo:dp}, which shows that \DP always finds an exact optimal solution. Moreover, we analyze the time and space complexity of Algorithm~\ref{algo:dp}.

\stitle{Correctness analysis}. We use the induction idea to prove the correctness of \DP algorithm. Consider a subtree $T_u$ rooted by $u$, we first assume that the optimal solution of its children $x \in \nb^{-}(u)$ in the sub-problem is $\DP(x, k, S)$ and the optimal selection is denoted by $S_u^{*,k}$ in all the following lemmas. Based on these lemmas, we can derive the theorem of algorithm correctness.

\begin{lemma} 
\label{lemma.yes}
Give a subtree $T_{u}$ rooted by $u$, a summary set $S$, which satisfies $S \cap \des(u) = \emptyset$ and a number $k$, we choose additional $k - 1$ summary vertices $S_{u}^{k - 1} \subseteq \des(u) \setminus \{u\}$. 
The largest summary score is $\g_{T_u}(S \cup \{u\} \cup S_{u}^{*,k - 1}) = \mathcal{Y}(u, k, S)$.
\end{lemma}

\begin{proof}
Assume $S_{u}^{*,k-1}$ is the optimal selection of summary subproblem on a subtree $T_u$. First, $\g_{T_u}(S \cup \{u\} \cup S_{u}^{k-1^*}) \geq  \mathcal{Y}(u, k, S)$ due to the optimal answer $S_{u}^{*,k-1}$. Next, we decompose the vertex set $S_{u}^{*,k-1}$ into multiple subsets $\bigcup_{x \in \nb^{-}(u)}{S_{x}^{*,k_x}}$ for each children $x\in \nb^{-}(u)$. In this way, $\g_{T_u}(S \cup \{u\} \cup S_{u}^{*,k-1}) = \feq(u) + \sum_{x \in \nb^{-}(u)}{\g_{T_x}(S \cup \{u\} \cup S_{x}^{*,k_x})} \leq \feq(u) + \sum_{x \in \nb^{-}(u)}{\DP(x, k_x, S \cup \{u\})} \leq \feq(u)+ \max\{\sum_{x\in N^-(u)} \DP(x, k_x, S\cup\{u\})\} = \mathcal{Y}(u, k, S)$. Thus, $\g_{T_u}(S \cup \{u\} \cup S_{u}^{*,k - 1}) = \mathcal{Y}(u, k, S)$ holds.
\end{proof}

\begin{lemma} 
\label{lemma.no}
Give a subtree $T_{u}$ rooted by $u$, a summary set $S$, which satisfies $S \cap \des(u) = \emptyset$ and a number $k$, we choose additional $k$ summary vertices $S_{u}^{k} \subseteq \des(u) \setminus \{u\}$. 
The largest summary score is $\g_{T_u}(S \cup S_{u}^{k^*}) = \mathcal{N}(u, k, S)$.
\end{lemma}

\begin{proof}
Assume $S_{u}^{*,k}$ is the optimal selection of summary subproblem on a subtree $T_u$. First, $\g_{T_u}(S \cup S_{u}^{*,k}) \geq  \mathcal{N}(u, k, S)$ due to the optimal answer $S_{u}^{*,k}$. Next,   we decompose the vertex set $S_{u}^{*,k}$ into multiple subsets $\bigcup_{x \in \nb^{-}(u)}{S_{x}^{*,k_x}}$. for each children $x\in \nb^{-}(u)$. As a result, $\g_{T_u}(S \cup S_{u}^{*,k}) = \smy_S(u) + \sum_{x \in \nb^{-}(u)}{\g_{T_x}(S \cup S_{x}^{*,k_x})} \leq \smy_S(u) + \sum_{x \in \nb^{-}(u)}{\DP(x, k_x, S)} \leq \smy_S(u)+ \max\{\sum_{x\in N^-(u)} \DP(x, k_x, S)\} = \mathcal{N}(u, k, S)$. Overall, $\g_{T_u}(S \cup S_{u}^{*,k}) = \mathcal{N}(u, k, S)$ holds.
\end{proof}

In the following, we prove the initialization cases of $\g_{T_u}(S \cup S_{u}^{k}) = \DP(u, k, S)$ for leaf nodes $u$ with $k=0$ and $k=1$.

\begin{lemma} 
\label{lemma.leaf}
Give a subtree $T_{u}$ rooted by a leaf node $u$, a summary set $S$, which satisfies $S \cap \des(u) = \emptyset$ and a number $k \leq 1$, we choose additional $k$ summary vertices $S_{u}^{k} \subseteq \des(u)$. The largest summary score is $\g_{T_u}(S \cup S_{u}^{*,k}) = \DP(u, k, S)$.
\end{lemma}

\begin{proof}
We consider two cases of $k=0$ and $k=1$. For $k = 1$, $\g_{T_u}(S \cup S_{u}^{*,k}) = \g_{T_u}(S \cup \{u\}) = \feq(u) = \DP(u, k, S)$; For $k = 0$, $\g_{T_u}(S \cup S_{u}^{*,k}) = \g_{T_u}(S) = \smy_{S}(u) = \DP(u, k, S)$. Therefore, $\g_{T_u}(S \cup S_{u}^{*,k}) = \DP(u, k, S)$ holds for leaf node $u$ with $k=1$ and $k=0$.
\end{proof}

\begin{theorem} \label{theorem.correct}
Give a tree $T$ rooted by $r$ and a number $k$, we choose $k$ summary vertices $S_r^{k}$ by algorithm~\ref{algo:dp}. The largest summary score is $\g(S_r^{*,k}) = \DP(r, k, \emptyset)$.
\end{theorem}

\begin{proof}
Based on Eq.~\ref{eq.max}, Lemma~\ref{lemma.yes} and Lemma~\ref{lemma.no}, we obtain $\DP(u, k, S) = \max\{\mathcal{Y}(u,k, S), \mathcal{N}(u,k,S)\} = \g_{T_u}(S \cup S_{u}^{*,k}) $. In addition, based on the induction idea and exact initialization cases in Lemma~\ref{lemma.leaf}, our answer of $\DP(r, k, \emptyset) = \g(\emptyset \cup S_r^{*,k}) = \g(S_r^{*,k})$ is the optimal solution. 
\end{proof}

\stitle{Complexity analysis}. The total number of states is $O(nhk)$ where $h$ is the height of $T$ and the transfer equation takes $|N^-(u)| \cdot k^2$. So, the time complexity of \DP in Algorithm~\ref{algo:dp} is $O(\sum_{u \in V} |N^-(u)| \cdot k^3 \cdot h)$ $\subseteq O(nhk^3)$. 
If $T$ is a complete binary tree with a height of $h \in O(\log n)$, it takes O($n \log n k^3$) time. Moreover, each state takes $O(k)$ memory to store the selected set, so the space complexity is $O(nhk^2)$.

%% file: tex/vtree.tex
\section{Tree Reduction for Fast Summarization}\label{sec.vtree}

In this section, we propose a tree reduction method \vtree to accelerate summarization process and achieve optimal answers. \vtree removes the useless nodes from tree $T$ and generate a small tree $T^*$ with $|T^*|\leq |T|$. We also analyze the correctness and complexity of \vtree. 

\stitle{Overview.} We observe that there exist several vertices could not be answer candidates.  We can remove useless vertices from tree $T$ based on the vertices with positive weights in the important node set $\mathcal{I}$. Specifically, those useless vertices satisfy two conditions  at the same time. First, each useless vertex $u$ has zero weights, i.e., $\feq(u)=0$. Second, each useless vertex $u$ is not a lowest common ancestor for any vertex subset $\mathcal{I}' \subseteq \mathcal{I}$. In this way, we can remove all these useless vertices from $T$ and generate a new tree $T^*(V^*, E^*)$, which significantly reduces the size of tree $T (V, E)$.  In many real applications, a large number of vertices have zero-weights in tree datasets as shown in Section~\ref{sec.exp}. 

To understand the correctness of our removal strategy, we show that Algorithm~\ref{algo:dp} achieves the same answers $S^*$ on original tree $T$ and reduced tree $T^*$. In other words, the useless vertices that are removed based on the above two conditions, do not appear in the answer $S^*$. Let us consider a deleted vertex $v\in V$ and $v \notin V^*$.   

First, there exists another vertex $u \in V^*$ no worse than $v$ for the summary answers. Alternatively, the representative impact score of $u$ is no less than the representative impact score of $v$ with regard to any vertex $x\in des(v)$, i.e., $\rep_{u}(x) \geq \rep_{v}(x)$. 
The rational reasons are as follows. \TKDE{Assume that an non-empty set $\mathcal{I}^\prime = \mathcal{I} \cap \des(v)$ and $u = \LCA(\mathcal{I}^\prime) \in V^*$. Based on $u$ is the least common ancestor, $v$ is an ancestor of $u$, so $\rep_{u}(x) \geq \rep_{v}(x)$ for each $x \in \mathcal{I}^\prime$. Hence, in whatever cases, $u$ is a better choice than $v$ in the subtree $T_v$. Next, we analyze the correctness of \DP algorithm in the new tree $T^*$. For each vertex $v \notin V^*$, it can be replaced by a better vertex $u \in V^*$. Thus, it wouldn't be selected as an answer in sub-problems by \DP. Thus, \DP finds the optimal answers in $T^*$ as in $T$.}

However, identifying all Lowest Common Ancestors (\LCA) \cite{bender2000lca} of any vertex subset $\mathcal{I}' \subseteq \mathcal{I}$, is very time consuming. Given three vertices $x, y, z$, we denote $\LCA(x, y, z) $ and $ \LCA(x, y)$ by the LCAs of vertices $\{x, y, z\}$ and $\{x, y\}$ respectively. We use the preorder traversal (Euler tour \cite{bender2000lca}) of a tree to optimize the \vtree algorithm. Similar as preorder traversal in a binary tree, preorder traversal in general tree traverses root firstly and then traverses the children from left to right. In the preorder traversal, we call $u$ is before $v$ if a vertex $u$ is traversed before $v$, denoted by $u < v$. We also define a sequenced list of vertices $P = \{v_1, v_2, ..., v_{|\mathcal{I}|} | v_i\in\mathcal{I}\}$ sorted by the preorder traversal. Based on the preorder ranking of vertices, we introduce the following lemma. 

\begin{lemma}
\label{lemma.lca}
For $\forall x, y, z \in V$, if $x < y < z$, we have one and only one of two following cases: either $\LCA(x, y, z) = \LCA(x, z) = \LCA(x, y)$ or  $\LCA(x, y, z) = \LCA(x, z) = \LCA(y, z)$ holds.
\end{lemma}

\begin{proof}
The proof can be similarly done as \cite{bender2000lca}.
\end{proof}

Based on the Lemma~\ref{lemma.lca}, for each subset $\mathcal{I}' \subseteq \mathcal{I}$, the $\LCA(\mathcal{I}')$ equals the \LCA of the first vertex $v_1'$ and last vertex $v_{|\mathcal{I}'|}'$ in the ordered set $P_{\mathcal{I}'}$, i.e. $\LCA(\mathcal{I}') = \LCA(v_1', v_{|\mathcal{I}'|}')$. It also equals the $\LCA$ of two neighbor vertices in ordered set $P_{\mathcal{I}}$, i.e. $\LCA(v_1', v_{|\mathcal{I}'|}') \in \{\LCA(v_1, v_2), ... , \LCA(v_{|\mathcal{I}| - 1}, v_{|\mathcal{I}|})\}$. So, instead of finding all the \LCAs of each subset of $I$, we can only identify the \LCAs of two neighbor vertices in the ordered set $P$. It reduces the number of LCA calculations from $O({|\mathcal{I}|}^2)$ to $O(|\mathcal{I}|)$. Thanks to it, we get the upper bound of the size of $T^*$ as follows.

\begin{theorem}
\label{theorem.treesize}
The tree size of $T^*$ is $|V^*| \leq 2|\mathcal{I}| + 1$.
\end{theorem}

\begin{proof}
Based on the optimization, the number of additional vertices is not greater than $|\mathcal{I}|$. So, $|V^*| \leq |\mathcal{I} \cup \{r\}| + |\mathcal{I}| = 2|\mathcal{I}| + 1$.
\end{proof}

\begin{algorithm}[t]
  \caption{Vtree}
  \label{algo:vtree}
  \begin{algorithmic}[1]
    \Require A tree $T = (V, E, \feq)$, an important node set $\mathcal{I}$.
    \Ensure A reduced tree $T^* = (V^*, E^*)$.
    \State $V^* \gets \mathcal{I} \cup \{r\}$;
    \State Apply the preorder traversal on tree $T$ \cite{bender2000lca};
    \State Obtain a sequenced list of  vertices $P = \{v_i\in\mathcal{I}\} $ where $v_i$ is visited earlier than $v_j$ in  preorder traversal for $i\leq j$;
    \For{$i \leftarrow 1$ to $|\mathcal{I}| - 1$}
      \State $V^* \gets V^* \cup \{LCA(v_i, v_{i+1})\}$ ;
    \EndFor
    \State $T^* = (V^*, connect(V^*, T))$;
    \State \Return{$T^*$}
  \end{algorithmic}
\end{algorithm}

\TKDE{
\stitle{Vtree Algorithm.} Algorithm~\ref{algo:vtree} shows the pseudo code of tree reduction method. First, we collect all vertices with positive weights and the root $r$ into a set $V^* = \mathcal{I} \cup \{r\}$ (line 1). 
Then, we apply preorder traversal on the tree $T$ by depth-first search (DFS) and sort vertices in the set $\mathcal{I}$(lines 2-3). 
Next, we construct a tree $T^*$ consisted of vertices in $V^*$. For a vertex $u$ with $\feq(u)=0$, $u$ can be added into $T^*$, only when $u$ is the \LCA of two neighbor vertices in preorder $\mathcal{I}$ (lines 4-5). 
Finally, we add edges between the nodes in $V^*$. We assign an edge weight between $u$ and $v$ as $|\l(u) - \l(v)|$ in $T$ (line 6). 
}

\begin{figure}[t]
\centering
{
\subfigure[Tree $T$]{
\label{example.vtree.1}
\includegraphics[width=0.30\linewidth]{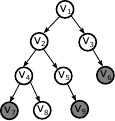} }
\quad
\subfigure[$v_1,v_2$ are LCA of $\mathcal{I}$]{
\label{example.vtree.2}
\includegraphics[width=0.30\linewidth]{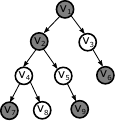} }
\quad
\subfigure[A new tree $T^*$]{
\label{example.vtree.3}
\includegraphics[width=0.23\linewidth]{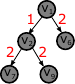} }
}
\caption{An example of tree reduction on $T$ with $\mathcal{I}=\{v_6, v_7, v_9\}$. 
A reduced tree $T^*$ of $T$ is shown in Figure~\ref{example.vtree.3}.}
\label{fig.vtree.example}
\vspace{-0.2cm}
\end{figure}

\begin{example}
Consider the tree $T$ shown in Figure~\ref{example.vtree.1}. The gray nodes have non-zero weights and belong to the important node set $\mathcal{I}=\{v_7, v_9, v_6\}$. The white nodes have zero weights, e.g., $v_1$, $v_2$, and so on. We apply the \vtree algorithm to reduce this tree $T$. The important steps are shown in Figure~\ref{fig.vtree.example}. The algorithm gets all the \LCA of the nodes following preorder traversal $v_7, v_9,$ and $v_6$. Thus, we consider two pairs of nodes, i.e., $(v_7, v_9)$ and $(v_9, v_6)$. First, it identifies the LCA of $v_7$ and $v_9$ as $v_2$, i.e., $\LCA(v_7, v_9) = v_2$, and then identifies $\LCA(v_9, v_6) = v_1$, which are colored in gray in Figure~\ref{example.vtree.2}. Next, \vtree removes from tree $T$ all nodes  that have zero-weights and are not identified as the qualified LCAs. It  adds edges between the remaining nodes in the new tree $T^*$ and assigns the corresponding edge weights in Figure~\ref{example.vtree.3}. The weight between $v_2$ and $v_7$ is $w(v_2, v_7)= 2$, due to that the level difference between $v_2$ and $v_7$ is 2 in the original tree $T$ in  Figure~\ref{example.vtree.1}. 
\end{example}

\stitle{Complexity analysis.}
Based on \cite{bender2000lca}, each \LCA  calculation  takes $O(\log{h})$ time and $O(n \log{h})$  space. Thus, the \vtree in Algorithm~\ref{algo:vtree}  takes $O(|\mathcal{I}|\log{h})$ time and $O(n \log{h})$ space. Based on Theorem~\ref{theorem.treesize}, the size of new tree $T^*$ is $|V^*| \leq 2 |\mathcal{I}| + 1$.  
As a result, \DP takes $O(|\mathcal{I}| hk^3 + |\mathcal{I}| \log{h})$ $\subseteq O(|\mathcal{I}| hk^3)$ time and $O(|\mathcal{I}| hk^2 + n\log{h})$ space. Note that \DP applied on the reduced tree $T^*$ achieves the same optimal solution as the original tree $T$, but runs much faster due to a smaller tree with $|\mathcal{I}| \ll n$ in practice. 

%% file: tex/exp.tex
\section{Experiments}\label{sec.exp}

In this section, we conduct extensive experiments to evaluate the performance of our proposed methods. All algorithms are implemented in C++. \TKDE{The source codes are publicly available.\footnote{\url{https://github.com/csxlzhu/TKDE_OTS}} 
}

\begin{table}[t]
\centering
\caption{The statistics of tree datasets.}\label{table:data}
\vspace{-0.3cm}
\scalebox{1.0}{
\begin{tabular}{c|c|c|c}
\toprule
Name& $n$& $|\mathcal{I}|$& Height\\
\midrule
\LATT& 4,226& 960& 22\\
\LNUR& 4,226& 771& 22\\
\ANIM& 15,135& 4,350& 11\\
\IMAGE& 73,298& 5,000& 20\\
\YAGO& 493,839& 10,000& 17\\
\bottomrule
\end{tabular}
}
\vspace{-0.5cm}
\end{table}

\stitle{Datasets. } We use five real-world datasets of hierarchical tree $T$ containing weighted terminologies. First, \LATT and \LNUR are extracted from the Medical Entity Dictionary (MED) \cite{jing2011graphical}. The tree contains 4,226 nodes. Each node represents a MED term. In addition, we use two datasets of $\mathcal{I}$, where the dataset \LATT contains the information about how physicians query online knowledge resources, and the other dataset \LNUR contains the query information of nurses. These two datasets contain 960 records and 771 records, respectively. Each record consists of a MED term with a frequency count of its occurrence, representing its node weight. 
The third dataset, \ANIM, is extracted from the ``Anime'' catalog in Wikipedia \cite{ANIM}. The \ANIM tree contains 
15,135 vertices. Each vertex represents the animation websites or superior categories.
The weight of a vertex is the number of page-views on this vertex catalog within one month. 
The fourth dataset, \IMAGE, is extracted from the Image-net \cite{imagenet}. The \IMAGE tree contains 
73,298 vertices.
Each \emph{synset} tag represents a vertex, whose id is given in the \emph{wnid} attribute of the tag. 
We choose 5,000 random catalogs containing images as the set $\mathcal{I}$. The frequency of each catalog is the number of images in the catalog. Last, the fifth dataset, \YAGO, is extracted from the ontology structure yagoTaxonomy in multilingual Wikipedias~\cite{yago, mahdisoltani2013yago3}. The \YAGO tree contains 493,839 taxonomy vertices. An edge represents that the child vertex is a ``subClassOf'' the parent vertex. The frequency of each taxonomy vertex is the number of objects from yagoTypes belong to this taxonomy.

\begin{figure*}[t]
\centering
{
\subfigure[\LATT]{
\label{fig.exp1_1_1}
\includegraphics[width=0.205\linewidth]{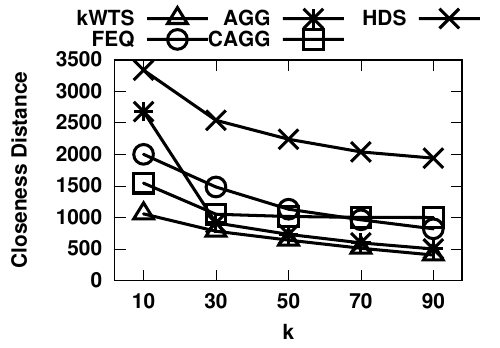} }\hskip -0.15in
\subfigure[\LNUR]{
\label{fig.exp1_1_2}
\includegraphics[width=0.205\linewidth]{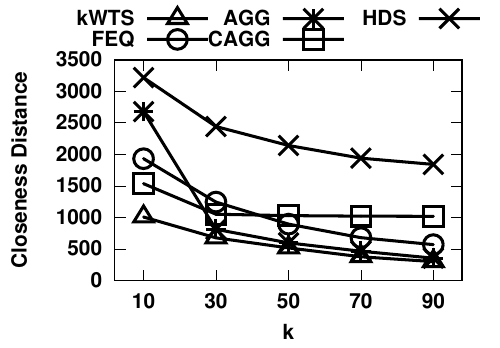} }\hskip -0.15in
\subfigure[\ANIM]{
\label{fig.exp1_1_3}
\includegraphics[width=0.205\linewidth]{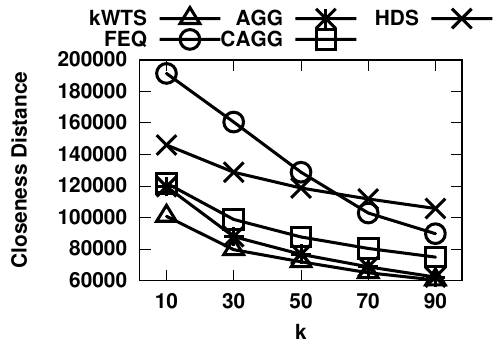} }\hskip -0.15in
\subfigure[\IMAGE]{
\label{fig.exp1_1_4}
\includegraphics[width=0.205\linewidth]{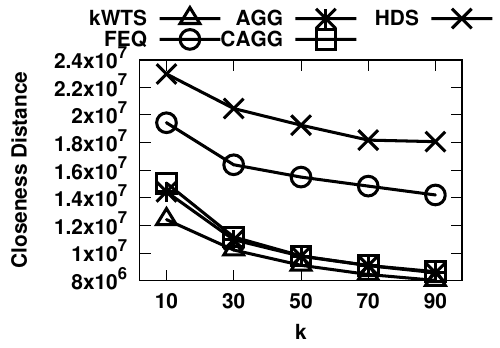} }\hskip -0.15in
\subfigure[\YAGO]{
\label{fig.exp1_1_5}
\includegraphics[width=0.205\linewidth]{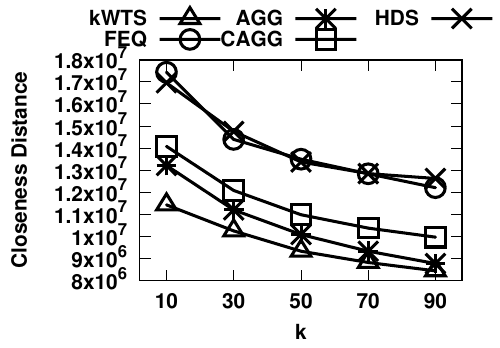} }
}
\vspace{-0.3cm}
\caption{Closeness distance of different models on all datasets.}
\label{fig.exp1_1}
\vspace{-0.4cm}
\end{figure*}

\begin{figure*}[t]
\centering
{
\subfigure[\LATT]{
\label{fig.exp1_2_1}
\includegraphics[width=0.205\linewidth]{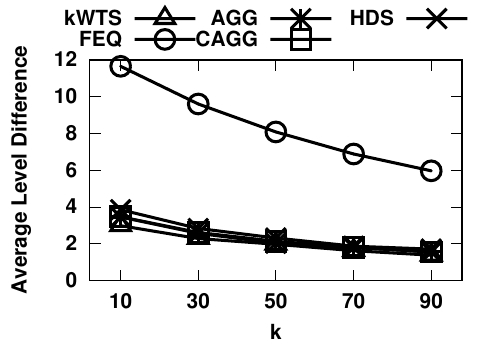} }\hskip -0.15in
\subfigure[\LNUR]{
\label{fig.exp1_2_2}
\includegraphics[width=0.205\linewidth]{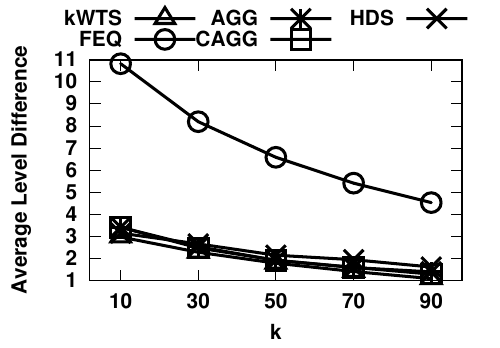} }\hskip -0.15in
\subfigure[\ANIM]{
\label{fig.exp1_2_3}
\includegraphics[width=0.205\linewidth]{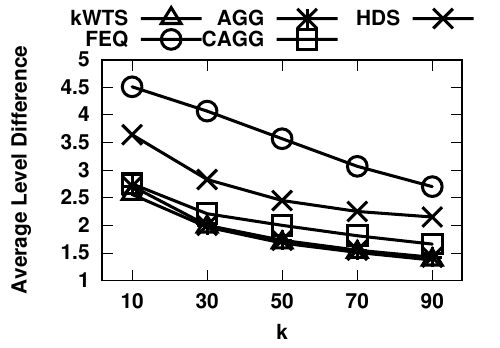} }\hskip -0.15in
\subfigure[\IMAGE]{
\label{fig.exp1_2_4}
\includegraphics[width=0.205\linewidth]{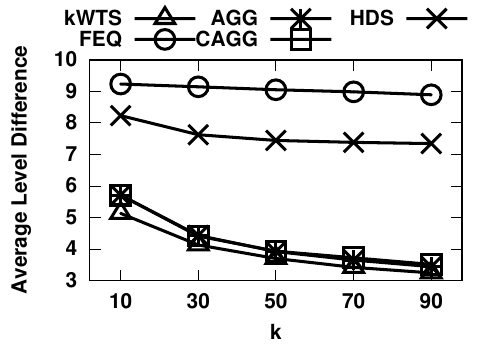} }\hskip -0.15in
\subfigure[\YAGO]{
\label{fig.exp1_2_4}
\includegraphics[width=0.205\linewidth]{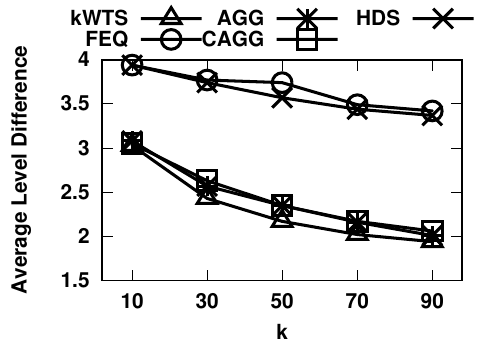} }
}
\vspace{-0.3cm}
\caption{Average level difference of different models on all datasets.}
\label{fig.exp1_2}
\vspace{-0.4cm}
\end{figure*}

\begin{figure*}[t]
\centering
{
\subfigure[\LATT]{
\label{fig.exp1_3_1}
\includegraphics[width=0.205\linewidth]{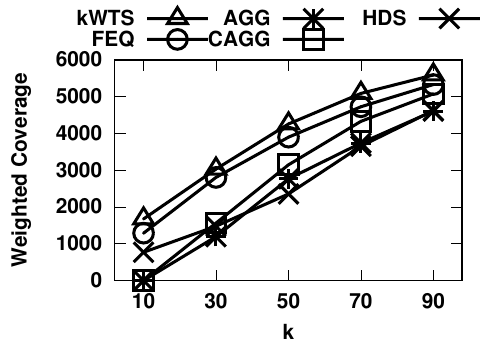} }\hskip -0.15in
\subfigure[\LNUR]{
\label{fig.exp1_3_2}
\includegraphics[width=0.205\linewidth]{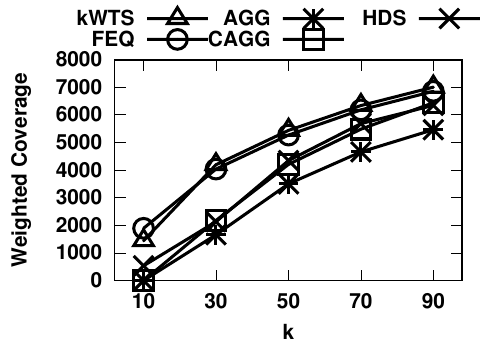} }\hskip -0.15in
\subfigure[\ANIM]{
\label{fig.exp1_3_3}
\includegraphics[width=0.205\linewidth]{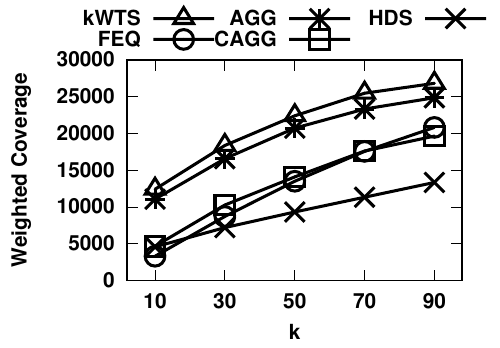} }\hskip -0.15in
\subfigure[\IMAGE]{
\label{fig.exp1_3_4}
\includegraphics[width=0.205\linewidth]{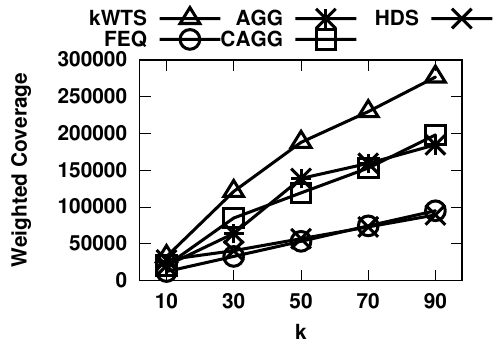} }\hskip -0.15in
\subfigure[\YAGO]{
\label{fig.exp1_3_5}
\includegraphics[width=0.205\linewidth]{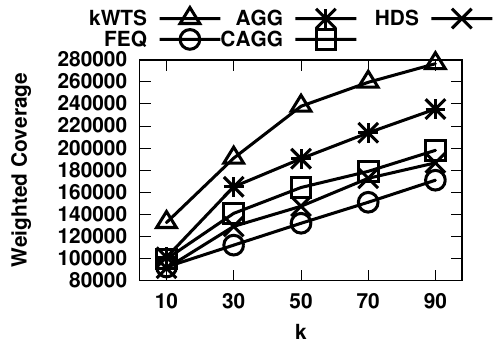} }
}
\vspace{-0.3cm}
\caption{Weighted coverage of different models on all datasets.}
\label{fig.exp1_3}
\end{figure*}

\stitle{Compared methods. }
To evaluate the effectiveness of our modeling problem \kVDOmodel, we evaluate and compare four competitive approaches 
-- \textbf{\FEQ}~\cite{jing2011graphical}, \textbf{\AGG}~\cite{jing2011graphical}, \textbf{\CAGG}~\cite{jing2011graphical}, and \textbf{\TS}~\cite{kimsummarizing}. 

\squishlisttight
\item  \textbf{\FEQ}: is a baseline approach, which selects $k$ nodes with the highest frequencies~\cite{jing2011graphical}.  

\item \textbf{\AGG}:  picks a set of $k$ nodes with the highest aggregate frequencies, where the aggregate frequency of a node $x$ is defined as $AF(v)=\sum_{y\in \dec(x)} \feq(y)$. 

\item  \textbf{\CAGG}: is a variant method of \textbf{\AGG} using another metric of contribution ratio. For a node $x$, the contribution ratio of $x$ is defined by $R(v)=\frac{AF(x)}{AF(y)}$ where $y$ is the parent of $x$. 
 Given a ratio threshold $\theta$, \CAGG selects the $k$ nodes that have the highest aggregate frequencies and the contribution ratio no less than $\theta$. We set $\theta=0.4$ by following \cite{jing2011graphical}. 
\item  \textbf{\TS}: is the state-of-the-art method of \underline{H}ierarchical \underline{D}ata \underline{S}ummaries, which creates a concise summary of similar weights in hierarchical multidimensional data~\cite{kimsummarizing}. To make a comparison, we use the 1-dimensional hierarchical data and a variant method to select exact $k$ summary vertices as the answer. 
\end{list}

Furthermore, we evaluate the effectiveness and efficiency of our algorithms \greedy (Algorithm~\ref{algo:greedy}) and \DP (Algorithm~\ref{algo:dp}), which respectively solve \kVDO approximately and optimally. We compare them with \Baseline greedy method and \brute. \Baseline greedy method gets the same solution as \greedy. Furthermore, \Baseline takes $O(n^2)$ time to compute $\triangle_{g}(x|S)$. Thus, the total time complexity of \Baseline is $O(n^3k)$. On the other hand, \brute achieves the same answer as \DP, which takes the exponential time w.r.t. the size of tree $|T|$ and  $k$.

\stitle{Evaluation metrics.}
To evaluate the quality of summary result $S$ found by all models, we use three metrics, i.e., the closeness distance $\CD(\mathcal{I}, S)$~\cite{huang2017ontology}, the average level difference $\ALD(\mathcal{I}, S)$~\cite{zhu2020top}, and also the weighted coverage $\WC(\mathcal{I}, S)$~\cite{zhu2020top, ranu2014answering}. 

\squishlisttight
\item[1.] The closeness distance $\CD(\mathcal{I}, S)$ is defined as the sum of weighted distance between $S$ to $\mathcal{I}$, denoted by 
$$\CD(\mathcal{I}, S) = \sum_{y\in \mathcal{I}} \min_{x\in S} \dist_{T}(x, y) \cdot \feq(y),$$
where $\dist_{T}(x, y)$ is 
the distance between $x$ and $y$ in  $T$. The smaller is $\CD(\mathcal{I}, S)$, the better is the summary quality.

\item[2.] The average level difference  $\ALD(\mathcal{I}, S)$ is a distance-based metric proposed in~\cite{zhu2020top}. $\ALD(\mathcal{I}, S)$ is defined as the average level difference between summary vertex and weighted vertex, denoted by 
$$\ALD(\mathcal{I}, S) =\frac{\sum_{y \in \mathcal{I}} \min_{x\in S \cap \anc(y)} (\l(y) - \l(x)) \cdot \feq(y)}{\sum_{y \in \mathcal{I}} \feq(y)}.$$
Note that we consider $\min_{x\in \emptyset} (\l(y) - \l(x)) = \l(y)$. 
The $\ALD(\mathcal{I}, S)$ metric takes into account both the level difference of summary results and the vertex weight. 
The smaller is $\ALD(\mathcal{I}, S)$, the better is the summary quality.

\item[3.] The weighted coverage $\WC(\mathcal{I}, S)$  is defined as the total weight of the vertices within summary set $S$ or their children, denoted by 
$$\WC(\mathcal{I}, S)=\sum_{x \in \mathcal{I} \cap C(S)}{\feq(x)},$$
where $C(S)= S \cup \bigcup_{x \in S} \nb^{-}(x)$. 
The $\WC(\mathcal{I}, S)$ metric evaluates the coverage of important vertices with large weights. The larger is $\WC(\mathcal{I}, S)$, the better is the result. 
\end{list}

Generally, both $\CD(\mathcal{I}, S)$ and $\ALD(\mathcal{I}, S)$ measure the distance between selection answer $S$ and important vertices $\mathcal{I}$. The smaller is the value, the better is the summary quality. Overall, these three evaluation metrics quantify the desiderata metrics of a good summarization in Section~\ref{sec.pre}. In addition, to evaluate the effectiveness of algorithms, we also use summary score $\g(S)$ to compare the results. The larger is $\g(S)$, the better is the solution. Furthermore, to evaluate the efficiency, we report the running time of different summarization algorithms. 
Note that we treat the running time as infinite if the algorithm run exceeds 3 hours.

\subsection{Effectiveness Evaluation}

\stitle{Exp-1: Quality comparison of different summarization models.} We compare the summarization quality of five different models \kVDOmodel, \FEQ, \AGG, \CAGG, and \TS. 
Figures \ref{fig.exp1_1}, \ref{fig.exp1_2} and \ref{fig.exp1_3} show the results of competitive methods on all real-world datasets, in terms of the closeness distance, the average level difference, and the weighted coverage, respectively. Note that we use \DP algorithm for \kVDOmodel model, which achieves the optimal solution. 
The size of summary set $k$ varies from $10$ to $90$. All models achieve smaller closeness distance, average level difference and larger weighted coverage with the increased $k$. 
On the other hand, our method \kVDOmodel can target these key vertices to obtain small closeness distances even when a small value  $k=10$, reflecting a superiority of \kVDOmodel against \AGG.  Figures~\ref{fig.exp1_3}(d) and~\ref{fig.exp1_3}(e) show that \kVDOmodel has much more substantial advantages than other methods, as the synthetic weighted nodes are more likely to locate at the bottom of tree in  IMAGE and YAGO. Furthermore, our model \kVDOmodel is a clear winner of all competitors, consistently achieving the smallest closeness distance, the smallest average level difference, and also the largest weighted coverage in Figures \ref{fig.exp1_1}, \ref{fig.exp1_2} and \ref{fig.exp1_3}. It significantly outperforms the other methods for a smaller $k$, which is a great help to shrink large datasets for tree summarization. 

\stitle{Remark.} Note that the closeness distance of \AGG has a sudden drop on \LATT and \LNUR from $k=10$ to $30$ as shown in Figures~\ref{fig.exp1_1}(a) and~\ref{fig.exp1_1}(b). This is caused by that the key vertices for tree summarization are selected by \AGG when a large summarization answer for $k=30$ but not a small answer for $k=10$. Correspondingly, the closeness distance of \AGG reduces significantly. Moreover,  \TS performs worse than \kVDOmodel in Figures \ref{fig.exp1_1}, \ref{fig.exp1_2} and \ref{fig.exp1_3}, due to that its objective is summarizing the changes between two trees but not exactly as our problem for a single tree. 
\revision{ }

\begin{table}[t]
\centering
\caption{Summary scores of \greedy and \DP, here $k = 25$.}\label{table.exp2}
\vspace{-0.3cm}
\scalebox{1.0}{
\begin{tabular}{c|c|c|c|c|c}
\toprule
Datasets& LATT& LNUR& ANIM& IMAGE& YAGO \\
\midrule
\greedy& 4,071& \textbf{5,048} & 18,628& 542,872& 1,492,101\\
\DP& \textbf{4,111} & \textbf{5,048} & \textbf{18,786}& \textbf{546,368}& \textbf{1,495,580}\\
\bottomrule
\end{tabular}
}
\end{table}

\stitle{Exp-2: Summary score comparison of greedy and optimal algorithms.}
We next conduct the effectiveness evaluation of our algorithm \greedy and \DP. Table~\ref{table.exp2} shows the summary scores of \greedy and \DP on five datasets. The exact algorithm \DP consistently outperforms \greedy on all datasets except \LNUR, verifying the effectiveness of our optimal solution against the greedy approach. 

\begin{figure}[t]
\centering
{
\subfigure[Summary score $\g(S)$]{
\label{fig.exp5_1}
\includegraphics[width=0.48\linewidth]{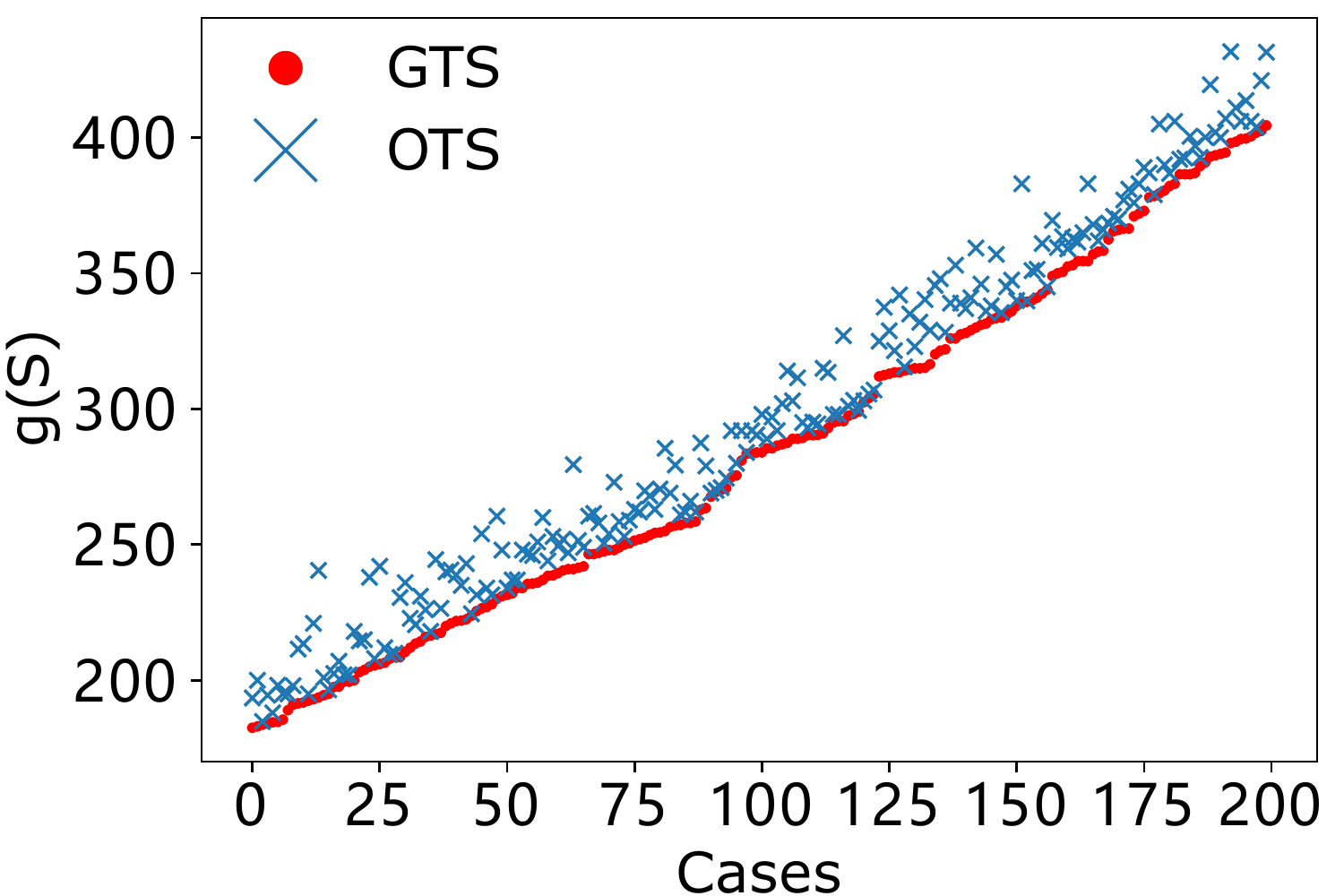} }
\subfigure[Running time(Seconds)]{
\label{fig.exp5_2}
\includegraphics[width=0.48\linewidth]{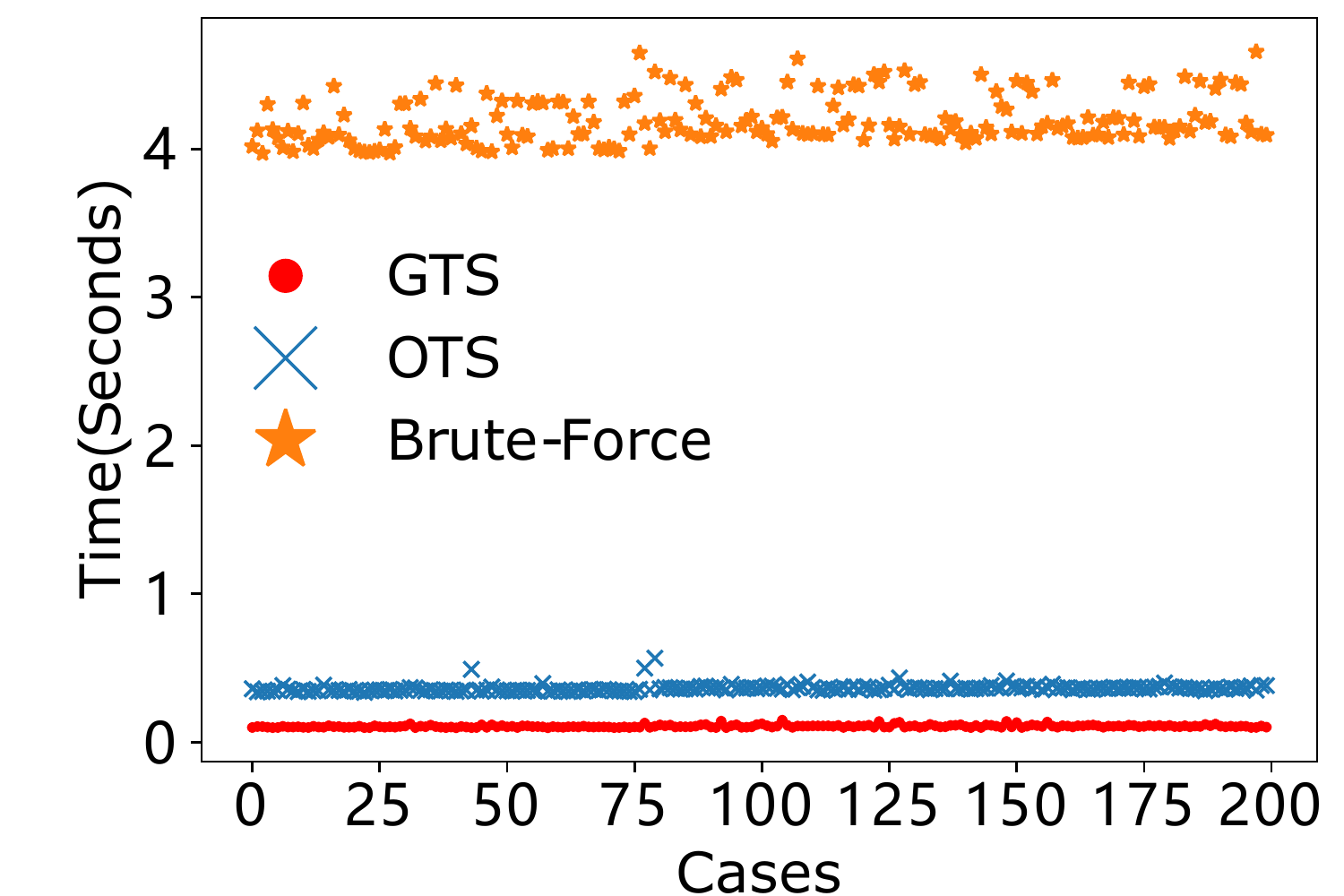} }
}
\vspace{-0.3cm}
\caption{Evaluation on 200 small synthetic datasets.}
\label{fig.exp5}
\end{figure}

\stitle{Exp-3: Approximation evaluation on small synthetic datasets.}
In this experiment, we evaluate the approximation of our algorithms w.r.t. the optimal answers. We randomly generate $200$ small-scale trees with $20$ nodes. 
We compare three methods of \greedy, \DP, and \brute. Note that \greedy produces no optimal solution in these cases. \DP and \brute always produce optimal answers.
Figure~\ref{fig.exp5_1} shows the summary score of three methods on $200$ cases. \DP gets the same solution of \brute, which verify the correctness of \DP.
As we can see that \DP wins the \greedy in all cases and \greedy achieves the average 95\%-approximation of optimal solutions. 
Figure~\ref{fig.exp5_2} shows the running time of three methods on all cases. \greedy and \DP run much faster than \brute, and \greedy is the winner. 

\subsection{Efficiency Evaluation}

\begin{figure*}[t]
\centering
{
\subfigure[\LATT]{
\label{fig.exp3_1}
\includegraphics[width=0.205\linewidth]{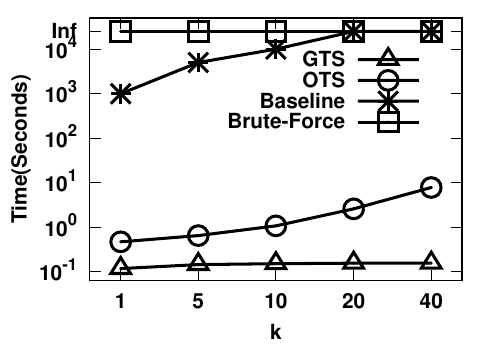} }\hskip -0.15in
\subfigure[\LNUR]{
\label{fig.exp3_2}
\includegraphics[width=0.205\linewidth]{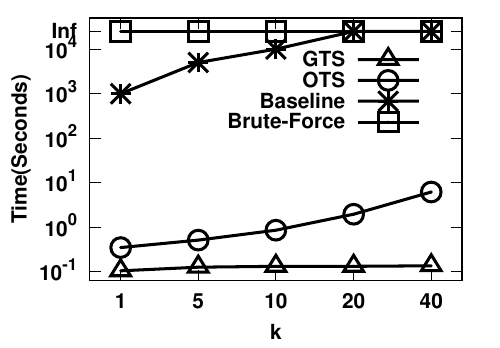} }\hskip -0.15in
\subfigure[\ANIM]{
\label{fig.exp3_3}
\includegraphics[width=0.205\linewidth]{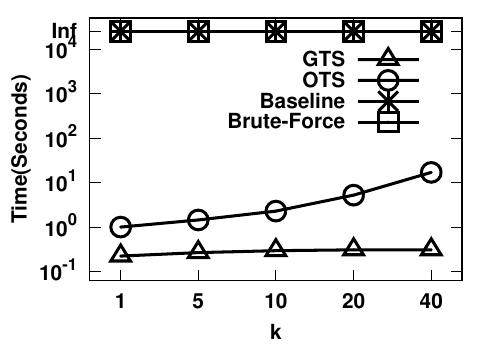} }\hskip -0.15in
\subfigure[\IMAGE]{
\label{fig.exp3_4}
\includegraphics[width=0.205\linewidth]{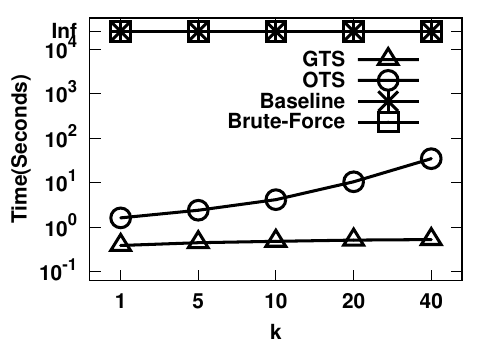} }\hskip -0.15in
\subfigure[\YAGO]{
\label{fig.exp3_5}
\includegraphics[width=0.205\linewidth]{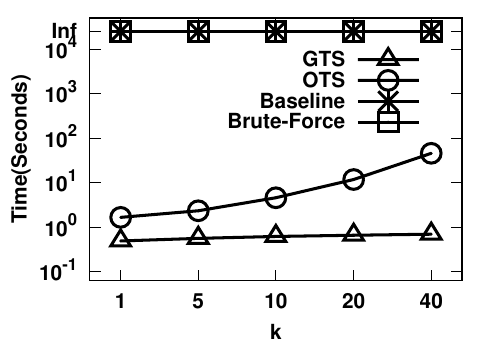} }
}
\vspace{-0.3cm}
\caption{Running time of different methods on all datasets.}
\label{fig.exp3}
\vspace{-0.5cm}
\end{figure*}

\stitle{Exp-4: Efficiency evaluation.}
We evaluate the running time of four methods \greedy, \DP, \Baseline, and \brute on all five datasets. \DP and \brute are optimal methods. \greedy and \Baseline are approximate methods. Figure~\ref{fig.exp3} shows running time of all methods when varying $k$. \greedy runs the fastest among them, which adopts an easy-to-compute greedy strategy. Interestingly, the efficiency of \DP is close to \greedy for small $k$ values. For the optimal methods, \DP runs much faster than \brute. Note that the running time results of \greedy and \DP are similar and scale well on large-scale tree datasets \IMAGE and \YAGO, as we invoke the \vtree algorithm to reduce the tree size of  \IMAGE and \YAGO into a very small one in $O(|\mathcal{I}|)$.

\begin{table}[t]
\centering
\caption{The size of new tree $T^*$ reduced by \vtree.}\label{table.exp4}
\vspace{-0.3cm}
\scalebox{1.0}{
\begin{tabular}{c|c|c|c|c|c}
\toprule
Datasets& LATT& LNUR& ANIM& IMAGE& YAGO \\
\midrule
$|T|$& 4,226& 4,226 & 15,135& 73,298& 493,839\\
$|\mathcal{I}|$& 960 & 771 & 4,350& 5,000& 10,000\\
$|T^*|$& 1,233 & 994 & 4,373& 6,402& 14,131\\
\bottomrule
\end{tabular}
}
\end{table}

\stitle{Exp-5: The size of reduced tree by Vtree.}
To verify the effectiveness of \vtree in Algorithm~\ref{algo:vtree}, we report the size of new trees $T^*$ reduced from $T$ on all real-world datasets. Table \ref{table.exp4} shows the size of original tree as $|T|$, the number of nodes with positive weights as $|\mathcal{I}|$, and the size of new tree $|T^*|$ by \vtree. The size of new tree $|T^*|$ is much smaller than the original tree size $|T|$. $|T^*|$ is also smaller than two times of $|\mathcal{I}|$, which confirms the results of Theorem~\ref{theorem.treesize}.

\begin{figure}[t]
\centering
{
\subfigure[Compute $\triangle_{g}(x|S)$]{
\label{fig.exp6_1}
\includegraphics[width=0.48\linewidth]{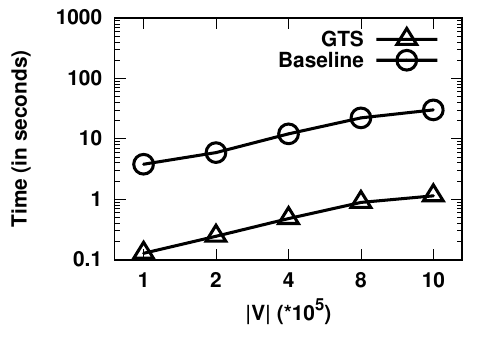} }
\subfigure[Tree summarization]{
\label{fig.exp6_2}
\includegraphics[width=0.48\linewidth]{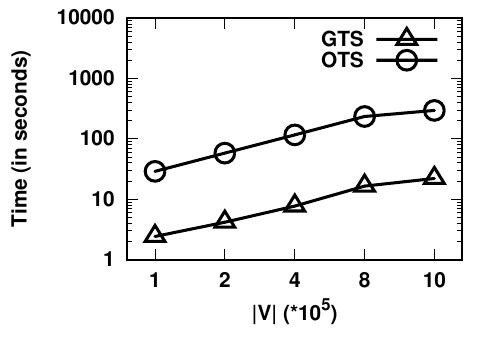} }
}
\caption{Scalability test on large synthetic datasets for $k=10$.}
\label{fig.exp6}
\end{figure}

\stitle{Exp-6: Scalability test.} In this experiment, we evaluate the scalability of \greedy and \DP by varying the size of tree $|V|$. We randomly generate 5 trees with size varying from $10^5$ to $10^6$, whose data statistics follow the real dataset \LATT.  We set the parameter $k = 10$. 
First, we test the scalability of computing  $\triangle_{g}(x|S)$. Note that the operation of  $\triangle_{g}(x|S)$ is to compute the marginal gain of summary scores, which is only used in greedy algorithms but not the global optimal \DP method. Thus, we only compare two greedy methods \Baseline and \greedy here. The running time results of computing  $\triangle_{g}(x|S)$ by Baseline and \greedy are shown in Figure \ref{fig.exp6_1}.
As we can see, \greedy is scalable very well with the increased size of tree nodes $|V|$.
Moreover, \greedy is much more efficient than \Baseline, which verifies the efficiency of fast computing $\triangle_{g}(x|S)$ in Algorithm~\ref{algo:gain}. 
Next, we evaluate the scalability of tree summarization by \greedy and \DP. 
 Figure~\ref{fig.exp6_2} reports the running time results on the increased tree datasets. As expected,  \greedy and \DP take longer time with the increasing $|V|$ stably, indicating that both methods scale well with a larger $|V|$.

\subsection{Case Study and Usability Evaluation}

\begin{figure}[t]
\centering
{
\subfigure{
\includegraphics[width=0.98\linewidth]{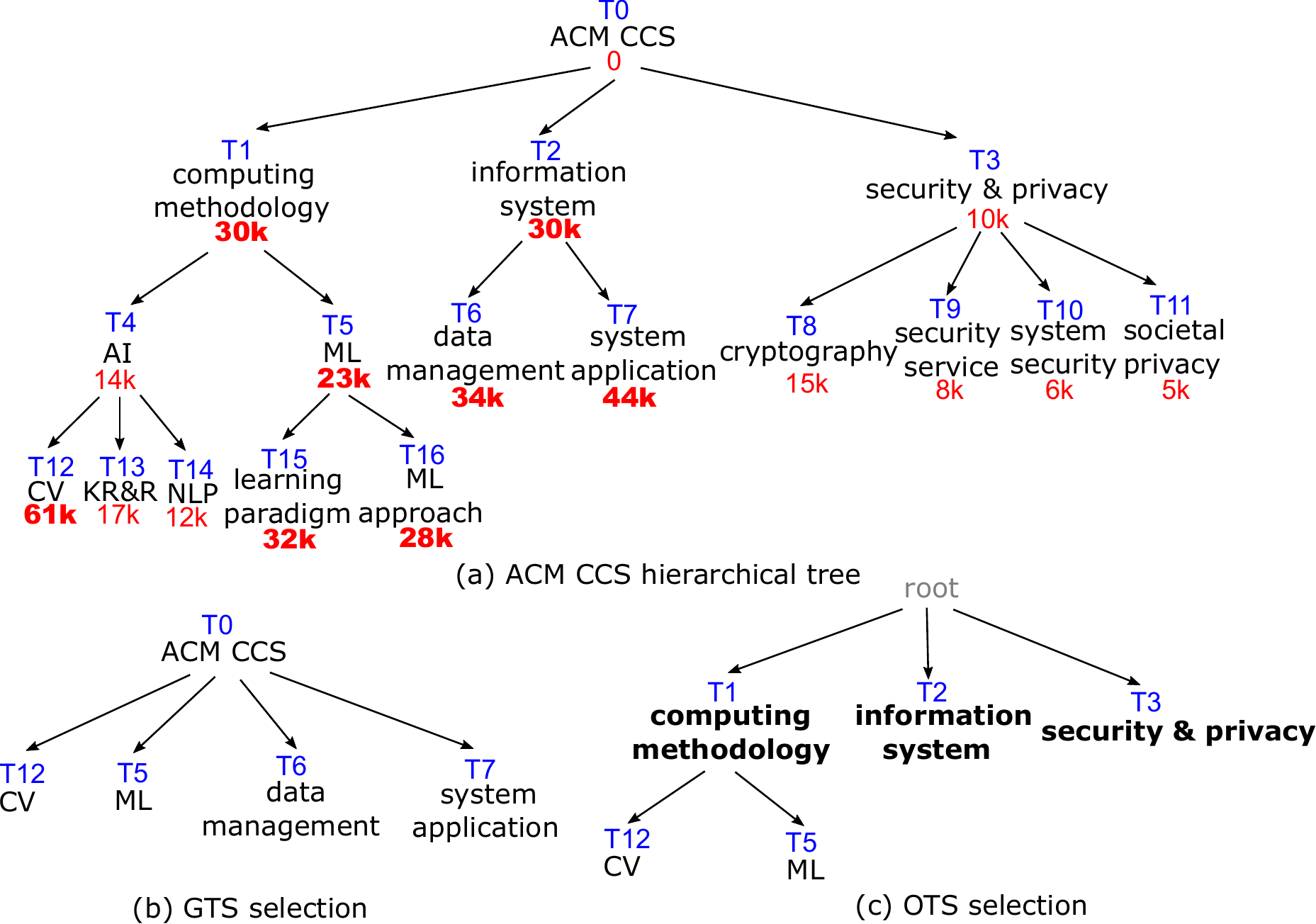} }
}
\vspace{-0.3cm}
\caption{A case study of tree summarization in ACM Computing Classification System dataset. The vertex weight is the cumulative number of papers published under the corresponding topic. The topic that has more than 20,000 papers, is depicted in red bold. Figure~\ref{fig.exp7}(b) and \ref{fig.exp7}(c) are graph visualization of top-$k$ summarization results by \greedy and \DP, respectively. Here, $k = 5$.}
\label{fig.exp7}
\end{figure}

\begin{figure}[t]
\centering
{
\subfigure{
\includegraphics[width=0.4\linewidth]{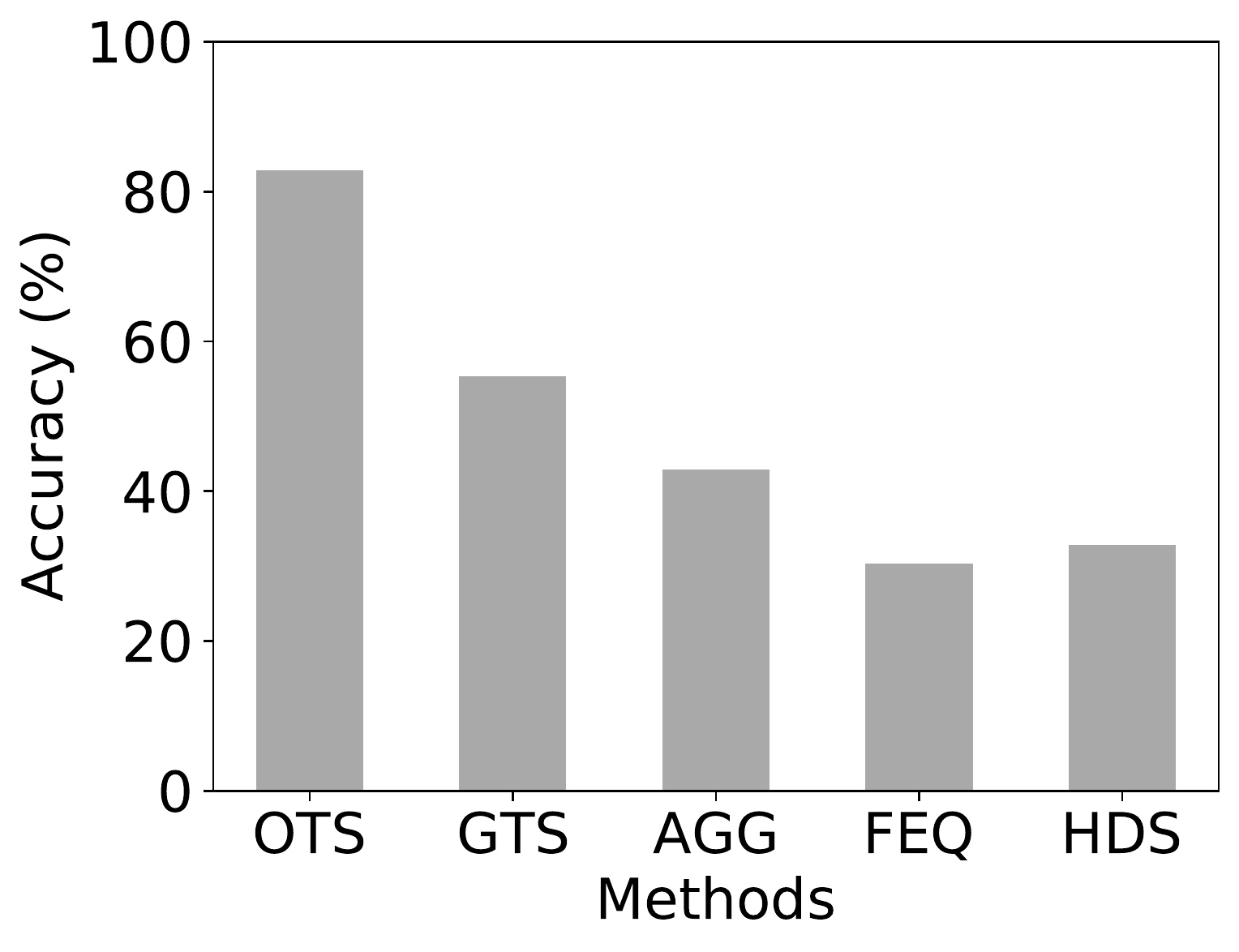} }
}
\vspace{-0.3cm}
\caption{Usability evaluation of different methods for top-$k$ attractive topic recommendation on ACM CCS dataset in Figure~\ref{fig.exp7}.}
\label{fig.exp7_3}
\vspace{-0.4cm}
\end{figure}

In this experiment, we conduct one case study and one  usability evaluation to validate the practical usefulness of our tree summarization model and algorithms. We construct a new real-world dataset with weighted terminologies from the ACM Computing Classification System (ACM CCS)~\cite{ACMCCS}, as shown in Figure~\ref{fig.exp7}(a). The hierarchical tree has 17 vertices, where each vertex represents a topic. An edge between two topics represents that the parent vertex is a generalized concept of its children topics; A child vertex is an instance concept of its parent topic. For example, the ``computing methodology'' topic (denoted as T1) generalizes two subcategories of `Artificial Intelligence' (denoted as T4, `AI') and `Machine Learning' (denoted as T5, `ML'). Moreover, each vertex is associated with a weight, representing the cumulative number of papers published under its topic in ACM DL~\cite{ACMCCS} from 2011 to 2021. For example, there exist 60,842 published papers related to `Computer Version' topic (denoted as T12, `CV'), i.e., $\feq(T12)=61\textbf{k}$. 
The higher the vertex weight, the more attractive the topic.

\stitle{Exp-7: Case study of summary topics in graph visualization.} We apply our summarization methods \greedy and \DP on the constructed ACM CCS dataset above. We set a small parameter $k=5$ to use only 5 topics to summarize the whole topic tree. 
Figures~\ref{fig.exp7}(b) and~\ref{fig.exp7}(c) show the topic selections of \greedy and \DP, respectively. \DP selects five vertices T1, T2, T3, T5, and T12, which cover three general attractive topics `computing methodology', `information system', `security \& privacy', and two attractive topics `ML' and `CV'. Note that we add a virtual root to connect with three vertices T1, T2, and T3, which follows the methodology of graph visualization in Section~\ref{sec.visualization}. On the other hand, \greedy selects a different answer of five vertices T0, T5, T6, T7, and T12 as shown in Figure~\ref{fig.exp7}(b). This greedy method always first selects T0 no matter what kinds of parameter setting on $k$. The greedy summary result in Figure~\ref{fig.exp7}(b) has one obvious shortcoming that cannot cover the topic of `security \& privacy' (denoted as T3). In addition, the summary score of \greedy is $253,079$, which is smaller than the answer of \DP with $\g(S) = 257,756$. Thus, the greedy selection is worse than the optimal answer in Figure~\ref{fig.exp7}(c), which has a more diverse coverage of different important topics and a larger summarization score.  

\stitle{Exp-8: Usability evaluation of summary topics.}  We conduct the usability evaluation for top-$k$ attractive topic recommendation, which selects $k$ topics to summarize attractive topics in ACM CCS dataset. We apply five methods \DP, \greedy, \AGG, \FEQ, and \TS on ACM CCS dataset. Note that \AGG and \CAGG select the same topics, thus only the topic selections of \AGG are reported here. We set $k=5$ and conduct a survey investigation.  Specifically, we ask 20 users, who are familiar with academic research and computer science topics. We request them to recommend top-5 most attractive topics of ACM CCS dataset in Figure~\ref{fig.exp7}(a). We evaluate an accuracy rate of matching topics between the users' choices and the methods' selections. Figure~\ref{fig.exp7_3} reports the average accuracy rates for all methods. Our method \DP achieves an accuracy rate of 82.5\%, which is the best performance among all methods. \greedy achieves the accuracy of 55\%, while other methods achieve no greater than an accuracy of 42.5\%. This usability evaluation validates the usefulness of our methods in attractive topics summarization on ACM CCS dataset.